\def\llncs{0}
\def\fullpage{1}
\def\anonymous{0}
\def\authnote{1}
\def\notxfont{0}
\def\submission{0}

\ifnum\submission=1
\def\llncs{1}
\fi

\ifnum\llncs=1 	\documentclass[envcountsect,a4paper,runningheads,10pt]{llncs}
\else
	\documentclass[letterpaper,hmargin=1.0in,vmargin=1.0in,11pt]{article}
			\ifnum\fullpage=1
		\usepackage{fullpage}
		\fi
\fi

\usepackage[%
  colorlinks=true,
  citecolor=blue,
  pagebackref=true
]{hyperref}

\usepackage{amsmath, amsfonts, amssymb, mathtools,amscd}

\usepackage{amsthm}

\usepackage{lmodern}
\usepackage[T1]{fontenc}
\usepackage[utf8]{inputenc}

\usepackage{arydshln} 
\usepackage{url}
\usepackage{ifthen}
\usepackage{bm}
\usepackage{multirow}
\usepackage[dvips]{graphicx}
\usepackage[usenames]{color}
\usepackage{xcolor,colortbl} 
\usepackage{threeparttable}
\usepackage{comment}
\usepackage{paralist,verbatim}
\usepackage{cases}
\usepackage{booktabs}
\usepackage{braket}
\usepackage{cancel} 
\usepackage{ascmac} 
\usepackage{framed}
\usepackage{authblk}
\usepackage{pifont}
\usepackage{qcircuit}
\usepackage{tikz}
\definecolor{darkblue}{rgb}{0,0,0.6}
\definecolor{darkgreen}{rgb}{0,0.5,0}
\definecolor{maroon}{rgb}{0.5,0.1,0.1}
\definecolor{dpurple}{rgb}{0.2,0,0.65}

\usepackage[capitalise,noabbrev]{cleveref}
\usepackage[absolute]{textpos}
\usepackage[final]{microtype}
\usepackage[absolute]{textpos}
\usepackage{everypage}
\DeclareMathAlphabet{\mathpzc}{OT1}{pzc}{m}{it}

\usepackage{algorithmic}
\usepackage{algorithm}
\usepackage{here}

\newtheoremstyle{thicktheorem}%
{\topsep}
{\topsep}
{\itshape}{}%
{\bfseries}%
{.}
{ }%
{\thmname{#1}\thmnumber{ #2}%
		\thmnote{ (#3)}%
}

\newtheoremstyle{remark}
{\topsep}
{\topsep}
	{}
	{}
	{}
	{.}
	{ }
	{\textit{\thmname{#1}}\thmnumber{ #2}
			\thmnote{ (#3)}%
	}

\ifnum\llncs=0
	\theoremstyle{thicktheorem}
	\newtheorem{theorem}{Theorem}[section]
	\newtheorem{lemma}[theorem]{Lemma}
	\newtheorem{corollary}[theorem]{Corollary}
	
	\newtheorem{definition}[theorem]{Definition}

	\theoremstyle{remark}
	
	\newtheorem{remark}[theorem]{Remark}

\else
\fi

\Crefname{MyClaim}{Claim}{Claims}

	\crefname{theorem}{Theorem}{Theorems}
	\crefname{assumption}{Assumption}{Assumptions}
	\crefname{construction}{Construction}{Constructions}
	\crefname{corollary}{Corollary}{Corollaries}
	\crefname{conjecture}{Conjecture}{Conjectures}
	\crefname{definition}{Definition}{Definitions}
	\crefname{exmaple}{Example}{Examples}
	\crefname{experiment}{Experiment}{Experiments}
	\crefname{counterexample}{Counterexample}{Counterexamples}
	\crefname{lemma}{Lemma}{Lemmata}
	\crefname{observation}{Observation}{Observations}
	\crefname{proposition}{Proposition}{Propositions}
	\crefname{remark}{Remark}{Remarks}
	\crefname{claim}{Claim}{Claims}
	\crefname{fact}{Fact}{Facts}
	\crefname{note}{Note}{Notes}

\ifnum\llncs=1
 \crefname{appendix}{App.}{Appendices}
 \crefname{section}{Sec.}{Sections}
\else
\fi

\ifnum\llncs=1
\renewcommand*{\backref}[1]{}
\else
	\renewcommand*{\backref}[1]{(Cited on page~#1.)}
	\ifnum\notxfont=1
	\else
		\usepackage{newtxtext}
	\fi
\fi
\ifnum\authnote=0  
\newcommand{\mor}[1]{}
\newcommand{\minki}[1]{}
\newcommand{\takashi}[1]{}

\else
\newcommand{\mor}[1]{$\ll$\textsf{\color{red} Tomoyuki: { #1}}$\gg$}
\newcommand{\takashi}[1]{$\ll$\textsf{\color{orange} Takashi: { #1}}$\gg$}

\fi

\newcommand{\SD}{\mathsf{SD}} 

\newcommand{\Good}{\mathsf{Good}}

\newcommand{\Samp}{\algo{Samp}}

\newcommand{\cA}{\mathcal{A}}
\newcommand{\cB}{\mathcal{B}}
\newcommand{\cC}{\mathcal{C}}
\newcommand{\cD}{\mathcal{D}}
\newcommand{\cE}{\mathcal{E}}

\def\makeuppercase#1{
\expandafter\newcommand\csname tl#1\endcsname{\widetilde{#1}}
}

\def\makelowercase#1{
\expandafter\newcommand\csname tl#1\endcsname{\widetilde{#1}}
}

\newcommand{\N}{\mathbb{N}}

\newcommand{\secp}{\lambda}

\newcommand*{\sk}{\keys{sk}}
\newcommand*{\pk}{\keys{pk}}

\newcommand{\ct}{\keys{ct}}

\newcommand*{\msg}{\keys{msg}}

\newcommand*{\keys}[1]{\mathsf{#1}}

\newcommand*{\algo}[1]{\ensuremath{\mathsf{#1}}}

\newenvironment{boxfig}[2]{\begin{figure}[#1]\fbox{\begin{minipage}{0.97\linewidth}
                        \vspace{0.2em}
                        \makebox[0.025\linewidth]{}
                        \begin{minipage}{0.95\linewidth}
            {{
                        #2 }}
                        \end{minipage}
                        \vspace{0.2em}
                        \end{minipage}}}{\end{figure}}

\newcommand{\bit}{\{0,1\}}

\newcommand{\Gen}{\algo{Gen}}

\newcommand{\Enc}{\algo{Enc}}
\newcommand{\Dec}{\algo{Dec}}

\newcommand{\Ver}{\algo{Ver}}

\newcommand{\Eval}{\algo{Eval}}

\newcommand{\Obf}{\mathsf{Obf}}
\newcommand{\obfC}{\hat{C}}
\newcommand{\CC}{\mathtt{C}}
\newcommand{\QQ}{\mathtt{Q}}
\newcommand{\Ext}{\mathsf{Ext}}

\newcommand{\negl}{{\mathsf{negl}}}

\newcommand{\poly}{{\mathrm{poly}}}

\makeatletter
\DeclareRobustCommand
  \myvdots{\vbox{\baselineskip4\p@ \lineskiplimit\z@
    \hbox{.}\hbox{.}\hbox{.}}}
\makeatother

\title{From Worst-Case Hardness of $\mathsf{NP}$ to Quantum Cryptography via Quantum Indistinguishability Obfuscation}

\ifnum\anonymous=1
\ifnum\llncs=1
\author{\empty}\institute{\empty}
\else
\author{}
\fi
\else
\ifnum\llncs=1
\author{
Tomoyuki Morimae\inst{1} \and Yuki Shirakawa\inst{1} \and Takashi Yamakawa\inst{2,3,1}
}
\institute{
 Yukawa Institute for Theoretical Physics, Kyoto University, Kyoto, Japan \and NTT Social Informatics Laboratories, Tokyo, Japan \and NTT Research Center for Theoretical Quantum Information, Atsugi, Japan
}
\else
\author[1]{Tomoyuki Morimae}
\author[1]{ Yuki Shirakawa}
\author[2,3,1]{ Takashi Yamakawa}
\affil[1]{{\small Yukawa Institute for Theoretical Physics, Kyoto University, Kyoto, Japan}\authorcr{\small tomoyuki.morimae@yukawa.kyoto-u.ac.jp} \authorcr{\small yuki.shirakawa@yukawa.kyoto-u.ac.jp}}
\affil[2]{{\small NTT Social Informatics Laboratories, Tokyo, Japan}\authorcr{\small takashi.yamakawa@ntt.com}}
\affil[3]{{\small NTT Research Center for Theoretical Quantum Information, Atsugi, Japan}}
\fi 
\fi

\date{}

\begin{document}

\maketitle

\begin{abstract}
Indistinguishability obfuscation (iO) has emerged as a powerful cryptographic primitive with many implications. While classical iO, combined with the infinitely-often worst-case hardness of $\mathsf{NP}$, is known to imply one-way functions (OWFs) and a range of advanced cryptographic primitives, the cryptographic implications of quantum iO remain poorly understood. In this work, we initiate a study of the power of quantum iO. We define several natural variants of quantum iO, distinguished by whether the obfuscation algorithm, evaluation algorithm, and description of obfuscated program are classical or quantum. For each variant, we identify quantum cryptographic primitives that can be constructed under the assumption of quantum iO and the infinitely-often quantum worst-case hardness of $\mathsf{NP}$ (i.e., $\mathsf{NP} \not\subseteq \mathsf{i.o.BQP}$). 
In particular, we construct pseudorandom unitaries, QCCC quantum public-key encryption and (QCCC) quantum symmetric-key encryption, 
and several primitives implied by them such as
one-way state generators, (efficiently-verifiable) one-way puzzles, and EFI pairs, etc.
While our main focus is on quantum iO, even in the classical setting, our techniques yield a new and arguably simpler construction of OWFs from classical (imperfect) iO and the infinitely-often worst-case hardness of $\mathsf{NP}$.
\end{abstract}

\thispagestyle{empty}
\newpage

\setcounter{tocdepth}{2}
\tableofcontents
\thispagestyle{empty}
\newpage

\setcounter{page}{1}

\section{Introduction}
\label{sec:introduction} 
Obfuscation enables us to transform a program into an unintelligible, ``scrambled'' form while preserving its functionality. The cryptographic study of obfuscation was initiated by Barak et al.~\cite{JACM:BGIRSVY12}, who formalized the notion and introduced the concept of virtual black-box (VBB) security. While they showed that VBB security is unachievable in general, they proposed a weaker notion called indistinguishability obfuscation (iO), leaving open the possibility of its realization.

About a decade later,\footnote{A preliminary version of \cite{JACM:BGIRSVY12} was published at CRYPTO 2001 and a preliminary version of \cite{SICOMP:GGHRSW16} was published at FOCS 2013.} Garg et al.~\cite{SICOMP:GGHRSW16} proposed the first candidate construction of iO. This breakthrough led to a cascade of results: it was soon discovered that iO, when combined with one-way functions (OWFs), enables the construction of a wide array of powerful cryptographic primitives (e.g., \cite{C:HohSahWat13,EC:GGGJKL14,STOC:SahWat14,C:BonZha14,AC:KomNaoYog14,C:ChuLinPas15,STOC:KopLewWat15,STOC:CHNVW16}), some of which had no known constructions prior to iO. As a result, iO has come to be viewed as a ``central hub''~\cite{STOC:SahWat14} in cryptography.

Given its remarkable utility, iO has attracted extensive research interest, both in terms of proposing new constructions (e.g., \cite{C:CorLepTib13,EC:BGKPS14,TCC:BraRot14,C:CorLepTib15,TCC:GenGorHal15,EC:BMSZ16,TCC:GMMSSZ16,EC:BDGM20}, and developing attacks (e.g., \cite{EC:CHLRS15,EC:HuJia16,C:MilSahZha16,EC:CheGenHal17,C:CHKL18,C:CCHKL19}). This line of work culminated in a landmark result by Jain, Lin, and Sahai~\cite{STOC:JaiLinSah21}, who gave a construction of iO based on long-studied and well-founded cryptographic assumptions.

Despite the tremendous power of iO when combined with OWFs, it is notable that iO alone does not yield any cryptographic primitives. For example, in a hypothetical world where $\mathsf{P} = \mathsf{NP}$ (or even $\mathsf{BPP} = \mathsf{NP}$), no cryptographic primitives can exist, yet iO still exists~\cite{FOCS:KMNPRY14}. This highlights the fact that, for iO to be cryptographically useful, one must at least assume the worst-case hardness of $\mathsf{NP}$.

In this context, Komargodski et al.~\cite{FOCS:KMNPRY14} showed that, assuming only the (infinitely-often) worst-case hardness of $\mathsf{NP}$ (i.e., $\mathsf{NP} \not\subseteq \mathsf{i.o.BPP}$),\footnote{Here, a language $L$ is in $\mathsf{i.o.BPP}$ if all $x\in L\cap\bit^n$ are correctly decided in probabilistic polynomial time for infinitely-many $n\in\mathbb{N}$.} one can already construct OWFs from iO. Once OWFs are obtained, combining them with iO yields public-key encryption (PKE) and many other advanced cryptographic primitives.

The construction of a OWF based on iO in \cite{FOCS:KMNPRY14} is quite simple, at least assuming the perfect correctness of the obfuscator.  
Given an obfuscator $\Obf$, one can define a function $f$ by
\begin{align}
f(r)\coloneqq\Obf(Z;r)
\end{align}
where $Z$ denotes the zero-function that outputs $0$ on all inputs, and the notation ``$;r$'' indicates that $\Obf$ uses randomness $r$. The authors showed that the one-wayness of $f$ follows from the assumption that $\mathsf{NP} \not\subseteq \mathsf{i.o.BPP}$.  They also demonstrated that even an imperfectly correct iO suffices to construct OWFs, although the construction in that case is more intricate.

As the above construction illustrates, a key requirement in \cite{FOCS:KMNPRY14} is that the obfuscator $\Obf$ is derandomizable, meaning that it behaves deterministically when given fixed randomness. In contrast, recent works~\cite{alagic2016quantumobfuscation,C:ABDS21,LATIN:BroKaz21,ITCS:BarMal22,STOC:BKNY23,STOC:ColGun24,cryptoeprint:2025/891} have considered quantum obfuscation, where $\Obf$ is a quantum algorithm that obfuscates either classical or quantum circuits.  In this setting, the implicit assumption of derandomizability in \cite{FOCS:KMNPRY14} breaks down: quantum algorithms inherently involve randomness due to measurement, and thus cannot be derandomized. As a result, the classical approach of \cite{FOCS:KMNPRY14} does not extend to quantum obfuscation. In fact, it appears unlikely that quantum iO implies OWFs, since a quantum obfuscator is intrinsically a randomized object, making it intuitively useless for constructing deterministic primitives such as OWFs. Nonetheless, recent works~\cite{Kre21,STOC:KQST23,kretschmer2024quantumcomputableonewayfunctionsoneway,STOC:LomMaWri24} have identified several cryptographic primitives that may exist even in the absence of OWFs, including 
one-way state generators (OWSGs)~\cite{C:MorYam22,TQC:MorYam24}, pseudorandom state generators (PRSGs)~\cite{C:JiLiuSon18}, pseudorandom unitaries (PRUs)~\cite{C:JiLiuSon18,cryptoeprint:2024/1652}, EFI pairs~\cite{ITCS:BCQ23}, efficiently-verifiable one-way puzzles (EV-OWPuzzs)~\cite{C:ChuGolGra24}, one-way puzzles (OWPuzzs)~\cite{STOC:KhuTom24}, etc. These primitives lie below OWFs in known implications, and their existence under weaker assumptions raises the possibility that quantum iO, together with quantum worst-case hardness of $\mathsf{NP}$, may suffice to construct them. This leads us to the central question of this work:
\begin{center}
\emph{Does quantum iO imply any quantum cryptographic primitive, assuming only the quantum worst-case hardness of $\mathsf{NP}$?}
\end{center}

\subsection{Our Result}
In this work, we show that quantum iO for classical circuits, when combined with the infinitely-often quantum worst-case hardness of $\mathsf{NP}$ (i.e., $\mathsf{NP} \not\subseteq \mathsf{i.o.BQP}$), implies a range of quantum cryptographic primitives. The specific primitives that can be constructed depend on which components of the assumed iO, such as the obfuscation algorithm, the evaluation algorithm, and the description of obfuscated circuit, are quantum and which remain classical. We elaborate on these distinctions and their implications below. 

To capture the various flavors of quantum iO, we formalize it as a pair of quantum polynomial-time (QPT) algorithms: an obfuscation algorithm $\Obf$ and an evaluation algorithm $\Eval$:
\begin{description}
\item[$\Obf(1^\secp,C)$:] The obfuscation algorithm takes the security parameter $1^\secp$ and a classical circuit $C$ as input and outputs an obfuscated encoding $\hat{C}$, which may be a quantum state.
\item[$\Eval(\hat{C}, x)$:] The evaluation algorithm takes an obfuscated encoding $\hat{C}$ and an input $x$, and outputs $C(x)$ (with overwhelming probability).
\end{description}
The security requirement is that, for any pair of functionally equivalent classical circuits $C_0$ and $C_1$, the obfuscations $\Obf(C_0)$ and $\Obf(C_1)$ must be computationally indistinguishable to any QPT distinguisher.

We consider several variants of quantum iO, distinguished by whether each component, namely, the obfuscator $\Obf$, the evaluator $\Eval$, and the obfuscated encoding $\hat{C}$, is quantum or classical. Specifically, for each $(\mathtt{X}, \mathtt{Y}, \mathtt{Z}) \in \{\QQ, \CC\}^3$, we define $(\mathtt{X}, \mathtt{Y}, \mathtt{Z})$-iO as follows:
\begin{itemize}
\item If $\mathtt{X} = \QQ$, then $\Obf$ is a quantum algorithm; if $\mathtt{X} = \CC$, then $\Obf$ is classical.
\item If $\mathtt{Y} = \QQ$, then $\Eval$ is a quantum algorithm; if $\mathtt{Y} = \CC$, then $\Eval$ is classical.
\item If $\mathtt{Z} = \QQ$, then the obfuscated encoding $\hat{C}$ is a quantum state; if $\mathtt{Z} = \CC$, then $\hat{C}$ is a classical string.
\end{itemize}

While there are eight possible combinations of $(\mathtt{X}, \mathtt{Y}, \mathtt{Z})$, not all are meaningful. In particular, the case $\mathtt{Z} = \QQ$ only makes sense when both $\mathtt{X} = \QQ$ and $\mathtt{Y} = \QQ$, since a classical obfuscator cannot generate a quantum state and a classical evaluator cannot take a quantum state as input. Accordingly, we focus on the five meaningful variants: $(\QQ,\QQ,\QQ)$,  $(\QQ,\QQ,\CC)$, $(\QQ,\CC,\CC)$, $(\CC,\QQ,\CC)$, and $(\CC,\CC,\CC)$.

We emphasize that, throughout, the circuit being obfuscated is always a classical circuit. This modeling choice only strengthens our results since our goal is to identify lower bounds of quantum iO.

Assuming $\mathsf{NP} \not\subseteq \mathsf{i.o.BQP}$, we show the following results (see \Cref{figure} for the summary of the results):
\begin{itemize}
    \item $(\QQ,\QQ,\QQ)$-iO implies IND-CPA secure quantum symmetric key encryption (QSKE), where the secret key is classical but the ciphertext is quantum. In particular, it implies OWSGs and EFI pairs.
     \item $(\QQ,\QQ,\CC)$-iO implies IND-CPA secure symmetric key encryption (SKE) in the quantum-computation classical-communication (QCCC) model, referred to as QCCC SKE,  where all communication is classical, but local computations, such as encryption and decryption, may be quantum.
    In particular, it implies EV-OWPuzz, OWPuzzs, OWSGs, QEFID pairs and EFI pairs.
      \item $(\QQ,\CC,\CC)$-iO implies IND-CPA secure public key encryption (PKE) in the QCCC  model, referred to as QCCC PKE. 
    In particular, it implies EV-OWPuzz, OWPuzzs, OWSGs, QEFID pairs and EFI pairs.
          \item $(\CC,\QQ,\CC)$-iO implies IND-CPA secure QCCC PKE and (post-quantum) OWFs. 
    In particular, it implies all Microcrypt primitives implied by PRUs.
              \item $(\CC,\CC,\CC)$-iO implies IND-CPA secure PKE and (post-quantum) OWFs. 
    In particular, it implies all Microcrypt primitives implied by PRUs.
\end{itemize}

We remark that the implication of $(\CC, \CC, \CC)$-iO can be obtained via a straightforward adaptation of the construction in \cite{FOCS:KMNPRY14}, as all components involved are classical, except that we consider quantum adversaries. 
Nonetheless, we believe that our proof is arguably simpler than that of \cite{FOCS:KMNPRY14}, which requires cascaded obfuscation, i.e., obfuscating an already obfuscated circuit, whereas our approach avoids it. 
In addition, an advantage of our approach is that it only requires obfuscation for 3CNF formulas, rather than for general classical circuits. 
This resolves the open problem left by \cite{FOCS:KMNPRY14}, namely, constructing OWFs from imperfectly correct iO under the assumption of worst-case hardness of $\mathsf{NP}$. We note, however, that this open problem has been recently resolved (in a stronger form that only requires witness encryption) by completely different techniques~\cite{STOC:HirNan24,cryptoeprint:2024/800}.

\usetikzlibrary{positioning} 
\usetikzlibrary{calc} 
\usetikzlibrary {quotes}
\tikzset{>=latex} 

\tikzstyle{mysmallarrow}=[->,black,line width=1.6]
\tikzstyle{myblackbotharrow}=[<->,black,line width=1.6]
\tikzstyle{myredbotharrow}=[<->,red,line width=1.6]
\tikzstyle{newarrow}=[->,red,line width=1.6]
\tikzstyle{newsinglearrow}=[->,red,line width=1.6]
\tikzstyle{carrow}=[->,red,line width=1.6]
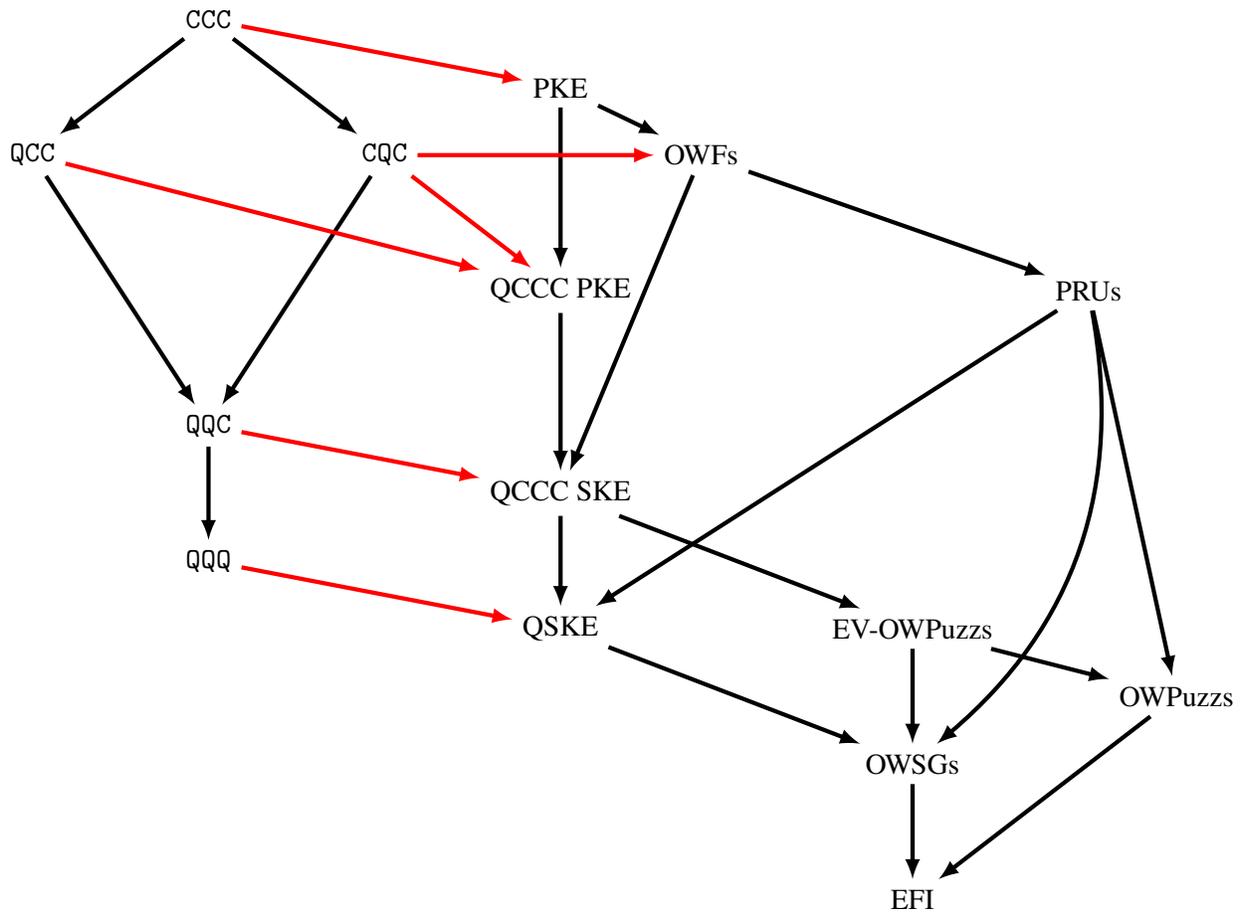
\begin{figure}
\begin{center}
    \begin{tikzpicture}[scale=0.9,every edge quotes/.style = {font=\footnotesize,fill=white}]
      \def\h{-2.0} 
      \def\w{2.6} 

      \node[] (CCC) at (1*\w,0*\h) {$\CC\CC\CC$};
      \node[] (QCC) at (0*\w,1*\h) {$\QQ\CC\CC$};
      \node[] (CQC) at (2*\w,1*\h) {$\CC\QQ\CC$};
      \node[] (QQC) at (1*\w,3*\h) {$\QQ\QQ\CC$};
      \node[] (QQQ) at (1*\w,4*\h) {$\QQ\QQ\QQ$};
      
      \node[] (PKE) at (3*\w,0.5*\h) {PKE};
      
      \node[] (OWF) at (3.8*\w,1*\h) {OWFs};
      
      \node[] (QCCCPKE) at (3*\w,2*\h) {QCCC PKE};
      \node[] (QCCCSKE) at (3*\w,3.5*\h) {QCCC SKE};
      \node[] (QSKE) at (3*\w,4.5*\h) {QSKE};
      
       \node[] (PRUs) at (6*\w,2*\h) {PRUs};
       \node[] (EVOWPuzz) at (5*\w,4.5*\h) {EV-OWPuzzs};
       \node[] (OWPuzz) at (6.5*\w,5*\h) {OWPuzzs};
       \node[] (OWSG) at (5*\w,5.5*\h) {OWSGs};
       \node[] (EFI) at (5*\w,6.5*\h) {EFI};

        \draw[mysmallarrow] (CCC) edge[] (QCC);
        \draw[mysmallarrow] (CCC) edge[] (CQC);
        \draw[mysmallarrow] (QCC) edge[] (QQC);
        \draw[mysmallarrow] (CQC) edge[] (QQC);
        \draw[mysmallarrow] (QQC) edge[] (QQQ);
        
        \draw[mysmallarrow,color=red] (CCC) edge[] (PKE);
        
        \draw[mysmallarrow] (PKE) edge[] (OWF);
        \draw[mysmallarrow] (PKE) edge[] (QCCCPKE);
        
        \draw[mysmallarrow,color=red] (QCC) edge[] (QCCCPKE);
        \draw[mysmallarrow,color=red] (CQC) edge[] (QCCCPKE);
        
        \draw[mysmallarrow,color=red] (CQC) edge[] (OWF);
        
        \draw[mysmallarrow] (QCCCPKE) edge[] (QCCCSKE);
        
        \draw[mysmallarrow,color=red] (QQC) edge[] (QCCCSKE);
        \draw[mysmallarrow] (OWF) edge[] (QCCCSKE);
        
        \draw[mysmallarrow] (QCCCSKE) edge[] (QSKE);
        \draw[mysmallarrow,color=red] (QQQ) edge[] (QSKE);

        \draw[mysmallarrow] (OWF) edge[] (PRUs);
        \draw[mysmallarrow] (QCCCSKE) edge[] (EVOWPuzz);
        \draw[mysmallarrow] (EVOWPuzz) edge[] (OWSG);
        \draw[mysmallarrow] (QSKE) edge[] (OWSG);
        \draw[mysmallarrow] (OWSG) edge[] (EFI);
        \draw[mysmallarrow][bend left=30] (PRUs) edge[] (OWSG);
        \draw[mysmallarrow] (PRUs) edge[] (QSKE);

        \draw[mysmallarrow] (EVOWPuzz) edge[] (OWPuzz);
        \draw[mysmallarrow] (OWPuzz) edge[] (EFI);
        \draw[mysmallarrow] (PRUs) edge[] (OWPuzz);
        
    \end{tikzpicture}
\end{center}
\caption{A summary of results (assuming $\mathsf{NP} \not\subseteq \mathsf{i.o.BQP}$).
Black arrows are known results or trivial implications.
Red arrows are new results. 
$\QQ\QQ\QQ$, for example, means the $(\QQ,\QQ,\QQ)$-iO.
}
\label{figure}
\end{figure}

\subsection{Technical Overview}
Throughout this technical overview, we assume the infinitely-often quantum worst-case hardness of $\mathsf{NP}$ (i.e., $\mathsf{NP} \not\subseteq \mathsf{i.o.BQP}$). Our starting point is the Valiant-Vazirani theorem~\cite{VV86}, which provides a randomized classical reduction from any $\mathsf{NP}$ instance to a $\mathsf{UP}$ instance.\footnote{$\mathsf{UP}$ is a subclass of $\mathsf{NP}$ consisting of problems where every yes-instance admits a unique witness.} This immediately implies that $\mathsf{NP} \not\subseteq \mathsf{i.o.BQP}$ implies $\mathsf{UP} \not\subseteq \mathsf{i.o.BQP}$ as well. Therefore, in proving cryptographic implications, we may assume $\mathsf{UP} \not\subseteq \mathsf{i.o.BQP}$ without loss of generality.

\paragraph{The case of $(\CC, \CC, \CC)$-iO.}
We begin by focusing on the case of $(\CC, \CC, \CC)$-iO, as the other cases build on similar ideas. In this setting, we construct a (post-quantum) OWF and a PKE scheme. Since a PKE scheme can be easily constructed from $(\CC, \CC, \CC)$-iO and OWFs using the same approach as in the classical construction of~\cite{STOC:SahWat14}, it suffices to focus on constructing OWFs.

Our main technical result is the following simple yet powerful statement. 
Let $\secp$ be the security parameter.
Let $m$ be a polynomial in $\secp$. For a string $k \in \bit^m$, let $P_k$ denote a circuit computing the \emph{point function} with target $k$, i.e., $P_k(k) = 1$ and $P_k(k') = 0$ for all $k' \ne k$. Let $Z_m$ be a circuit computing the \emph{zero function} on $m$-bit inputs, that is, $Z_m(k') = 0$ for all $k' \in \bit^m$. Then, for some polynomial $m$, 
assuming 
$\mathsf{NP} \not\subseteq \mathsf{i.o.BQP}$, 
we show that
\begin{align}
\Obf(1^\secp, P_k) \approx_c \Obf(1^\secp, Z_m),
\end{align}
where $k \gets \bit^m$,\footnote{Here $k\gets\bit^m$ means that $k$ is sampled uniformly at random from $\bit^m$.} and $\approx_c$ means computational indistinguishability against QPT distinguishers.\footnote{In the actual theorem, we account for the circuit sizes of $P_k$ and $Z_m$, but omit these details here for simplicity.}

This indistinguishability directly implies the existence of OWFs. Intuitively, the distributions $\Obf(1^\secp, P_k)$ and $\Obf(1^\secp, Z_m)$ should be \emph{statistically far}, since the obfuscation of $P_k$ is functionally equivalent to a point function, while the obfuscation of $Z_m$ is functionally equivalent to a constant-zero function. Moreover, since $\Obf$ is assumed to be classical, both distributions are classically efficiently samplable. Such pairs of classically efficiently samplable, statistically far, but computationally indistinguishable distributions are known as \emph{EFID pairs}~\cite{Gol90,ITCS:BCQ23}, and existentially equivalent to OWFs.\footnote{The proof of this fact in \cite{Gol90} only considers classical adversaries, but it extends to the post-quantum setting in a straightforward manner.}

We now describe the idea for proving the computational indistinguishability between $\Obf(1^\secp, P_k)$ and $\Obf(1^\secp, Z_m)$.  
As discussed above, we may assume $\mathsf{UP} \not\subseteq \mathsf{i.o.BQP}$. 
Therefore, to prove the above indistinguishability under the assumption, it suffices to show that any distinguisher between $\Obf(1^\secp, P_k)$ and $\Obf(1^\secp, Z_m)$ can be used to recover the unique witness of a yes-instance of a $\mathsf{UP}$ problem.

Let $x$ be a yes-instance of a $\mathsf{UP}$ language with a unique witness $w \in \bit^m$, and let $M$ be the corresponding verification algorithm. Our goal is to recover $w$ given $x$ and $M$. To do so, we take a uniformly random $v \in \bit^m$ and define a circuit $V[x,v](z)$ that outputs $1$ if $M(x, v \oplus z) = 1$ and otherwise outputs $0$. Note that $V[x,v](z)$ outputs $1$ only when $z = w \oplus v$, so it is functionally equivalent to the point function $P_{w \oplus v}$. Hence, by the security of iO, 
for each $v\in\bit^m$, 
we have
\begin{align}
\Obf(1^\secp, V[x,v]) \approx_c \Obf(1^\secp, P_{w \oplus v}).
\end{align} 
Moreover, since $v$ is chosen uniformly at random, $w \oplus v$ is also uniformly random, and so
\begin{align}
\Obf(1^\secp, P_{w \oplus v}) \equiv \Obf(1^\secp, P_k),
\end{align}
where $\equiv$ denotes distributional equivalence and $v,k \gets \bit^m$. Therefore, a distinguisher that distinguishes 
$\Obf(1^\secp, P_k)$ with $k\gets\bit^m$ from $\Obf(1^\secp, Z_m)$ also distinguishes $\Obf(1^\secp, V[x,v])$ with $v\gets\bit^m$ from $\Obf(1^\secp, Z_m)$.

While $V[x,v]$ and $Z_m$ are not functionally equivalent, they differ on only a single input, namely $w \oplus v$. It is known that any iO also satisfies the notion of \emph{differing-inputs obfuscation (diO)}~\cite{JACM:BGIRSVY12,EPRINT:ABGSZ13,TCC:BoyChuPas14} in the single differing-input setting. This means that if one can distinguish two circuits that differ on only one input, then one can efficiently extract that input. Thus, a distinguisher between $\Obf(1^\secp, V[x,v])$ with $v\gets\bit^m$ and $\Obf(1^\secp, Z_m)$ can be used to extract $w \oplus v$.

Since $v$ is chosen independently of both $x$ and $w$, we can construct a reduction algorithm that chooses $v$ on its own, extracts $w \oplus v$, and thereby recovers $w$. This completes the proof of the computational  indistinguishability.

\paragraph{The case of $(\CC, \QQ, \CC)$-iO.}
In this setting, we construct a (post-quantum) OWF and a QCCC PKE scheme. 

To construct a OWF, we use the same argument as in the $(\CC, \CC, \CC)$ setting. Notably, that construction does not rely on the assumption that $\Eval$ is classical, so the proof remains valid even when $\Eval$ is quantum.\footnote{Interestingly, the proof remains valid even if $\Eval$ is inefficient.}

For constructing a QCCC PKE scheme, we again follow the classical construction of~\cite{STOC:SahWat14}. The only difference is that evaluating the obfuscated program now involves quantum computation, making the encryption algorithm quantum and thus yielding a QCCC PKE scheme.

\paragraph{The case of $(\QQ, \CC, \CC)$-iO.}
In this setting, we construct a QCCC PKE scheme. Since $\Obf$ cannot be derandomized in this setting, it is unlikely that OWFs can be constructed. Therefore, we need a different approach from the classical construction of~\cite{STOC:SahWat14}.

We begin by noting that the following indistinguishability still holds in this setting:
\begin{align}
\Obf(1^\secp, P_k) \approx_c \Obf(1^\secp, Z_m),
\end{align}
where $k \gets \bit^m$. 
Based on this, we construct a QCCC PKE scheme as follows:
\begin{itemize}
    \item \textbf{Key generation:} Choose $k \gets \bit^m$, and output the public key $\hat{P}_k \gets \Obf(1^\secp, P_k)$ and the secret key $k$.
    
    \item \textbf{Encryption:} On input a public key $\hat{P}_k$ and a message $\msg$, define a classical circuit $C[\hat{P}_{k}, \msg]$ that takes $k' \in \bit^m$ as input and outputs $\msg$ if $\Eval(\hat{P}_{k}, k') = 1$, and outputs $\bot$ otherwise. Here, we assume for simplicity that $\Eval$ is deterministic (see the full proof in \Cref{{sec:QCC}} for how to handle randomized $\Eval$). The ciphertext is defined as $\hat{C}[\hat{P}_{k}, \msg] \gets \Obf(1^\secp, C[\hat{P}_{k}, \msg])$.
    
    \item \textbf{Decryption:} On input a ciphertext $\hat{C}[\hat{P}_{k}, \msg]$ and the secret key $k$, evaluate the obfuscated circuit on input $k$ and output the result $\msg'$.
\end{itemize}

The correctness of the scheme follows directly from the correctness of the iO.

\medskip

We now argue IND-CPA security of the scheme. First, consider a hybrid where the public key is replaced with $\Obf(1^\secp, Z_m)$. This is computationally indistinguishable from the real scheme by the indistinguishability between $\Obf(1^\secp, P_k)$ and $\Obf(1^\secp, Z_m)$, which has already been established.

Next, consider a second hybrid where the ciphertext is replaced with $\Obf(1^\secp, Z_m)$. This is indistinguishable from the previous hybrid because, when the public key is $\Obf(1^\secp, Z_m)$, the circuit $C[\Obf(1^\secp, Z_m), \msg]$ is functionally equivalent to the zero function regardless of the message $\msg$, and hence its obfuscation is indistinguishable from that of $Z_m$ by the security of iO.

In the final hybrid, the ciphertext reveals no information about the message $\msg$, so the scheme satisfies IND-CPA security.

\paragraph{The case of $(\QQ, \QQ, \CC)$-iO.}
In this setting, we construct a QCCC SKE scheme.

The idea is quite simple. We construct a QCCC SKE scheme for single-bit messages as follows: let $k \gets \bit^m$ be the secret key. To encrypt the message $0$, output the ciphertext $\Obf(1^\secp, Z_m)$; to encrypt the message $1$, output $\Obf(1^\secp, P_k)$. Decryption is performed by evaluating the obfuscated program (i.e., the ciphertext) on input $k$ and outputting the result, which will be either $0$ or $1$ accordingly.

While we have already established the indistinguishability between $\Obf(1^\secp, Z_m)$ and $\Obf(1^\secp, P_k)$, this alone is not sufficient to guarantee IND-CPA security of the above scheme. However, by carefully examining the proof of this indistinguishability, one can see that it extends to the case where the distinguisher is given multiple samples from the respective distributions, all generated using the same secret key $k$. This extension immediately implies the IND-CPA security of the scheme.

\paragraph{The case of $(\QQ, \QQ, \QQ)$-iO.}
In this setting, we construct a QSKE scheme, where the secret key is classical but the ciphertext is quantum.  The construction is exactly the same as in the $(\QQ, \QQ, \CC)$ case described above. The only difference is that the output of $\Obf$ is a quantum state, which means the ciphertexts are quantum. As a result, the scheme realizes QSKE rather than QCCC SKE.

\subsection{Related Work} 
Alagic and Fefferman \cite{alagic2016quantumobfuscation} and Alagic, Brakerski, Dulek, and Schaffner \cite{C:ABDS21} showed the impossibility of VBB obfuscation for classical circuits even when the obfuscator and the obfuscated programs are allowed to be quantum. 
Broadbent and Kazmi \cite{LATIN:BroKaz21} constructed an iO for quantum circuits, whose efficiency depends exponentially on the number of $\mathsf{T}$ gates in the circuit being obfuscated. Bartusek and Malavolta~\cite{ITCS:BarMal22} constructed 
iO for null quantum circuits in the classical oracle model.\footnote{A null quantum circuit is a quantum circuit with classical input and output that outputs $0$ with high probability on all inputs.} Bartusek, Kitagawa, Nishimaki, and Yamakawa~\cite{STOC:BKNY23} extended it to iO for pseudo-deterministic quantum circuits in the classical oracle model.\footnote{A pseudo-deterministic quantum circuit is a quantum circuit with classical input and output that computes a deterministic function with high probability.} 
Coladangelo and Gunn \cite{STOC:ColGun24} introduced the notion of quantum state iO, which allows for obfuscating a quantum description of a classical function, and provided a construction in the quantum oracle model. This was later improved by Bartusek, Brakerski, and Vaikuntanathan~\cite{STOC:BarBraVai24}, who gave a construction in the classical oracle model. Huang and Tang~\cite{cryptoeprint:2025/891} further improved it to support obfuscation of unitary quantum programs with quantum inputs and outputs.

Khurana and Tomer ~\cite{STOC:KhuTom25} presented a potential approach to constructing OWPuzzs based on the worst-case quantum hardness of $\#\mathsf{P}$. 
While the worst-case hardness of  $\#\mathsf{P}$ is a significantly weaker assumption than that of $\mathsf{NP}$, their current result relies on certain unproven conjectures related to quantum supremacy. While our work also relies on an additional assumption that quantum iO exists, our assumption is cryptographic in nature, whereas theirs pertains to quantum supremacy. Given the fundamental difference between the two, the approaches are not directly comparable.

Hirahara and Nanashima~\cite{STOC:HirNan24} proved that the infinitely-often worst-case (classical)  hardness of $\mathsf{NP}$ and (classically-secure) zero-knowledge arguments for $\mathsf{NP}$ imply OWFs (see also a simplified exposition in \cite{cryptoeprint:2024/800}). Since iO implies zero-knowledge arguments for $\mathsf{NP}$, this result improves upon that of~\cite{FOCS:KMNPRY14}. An interesting direction for future work is to establish a quantum analog of their result—for example, constructing quantum cryptographic primitives under the assumption of the infinitely-often worst-case quantum  hardness of $\mathsf{NP}$ and the existence of quantum zero-knowledge arguments for $\mathsf{NP}$.
Note that such a result is known if we assume the \emph{average-case} quantum hardness of $\mathsf{NP}$ instead of the worst-case quantum hardness~\cite{ITCS:BCQ23}.

A similar technique to ours—reducing an $\mathsf{NP}$ instance to a $\mathsf{UP}$ instance via the Valiant-Vazirani theorem, then rerandomizing the $\mathsf{UP}$ verification circuit through the obfuscation of a circuit in which the input is shifted by a random string—also appears in the work of Brakerski, Brzuska, and Fleischhacker~\cite{C:BraBrzFle16}, albeit in a completely different context. Their motivation was to prove the impossibility of statistically secure iO with approximate correctness. 

A concurrent work by Ilango and Lombardi~\cite{cryptoeprint:2025/1087} employs this technique in a context more closely related to ours. In particular, they independently present an alternative proof of the result in~\cite{FOCS:KMNPRY14} using this idea. Nevertheless, the focus of their work is quite different:  it is primarily set in the classical setting, aiming to establish fine-grained worst-case to average-case reductions using iO. While they do include one quantum result—providing proofs of quantumness from a worst-case assumption and iO—it relies on classical iO, and their work does not address quantum iO.
\section{Preliminaries}\label{sec:preliminaries}
\paragraph{Notations.} 
We use standard notations of quantum computing and cryptography.
We use $\secp$ as the security parameter.
$[n]$ means the set $\{1,2,...,n\}$.
For a finite set $S$, $x\gets S$ means that an element $x$ is sampled uniformly at random from the set $S$.
$\negl$ is a negligible function, and $\poly$ is a polynomial.
PPT stands for (classical) probabilistic polynomial-time and QPT stands for quantum polynomial-time. We refer to a non-uniform QPT algorithm as a QPT algorithm with polynomial-size quantum advice. 
We stress that the running time of the algorithm can be polynomial in $\secp$ rather than in $\log \secp$. 
For an algorithm $A$, $y\gets A(x)$ means that the algorithm $A$ outputs $y$ on input $x$.

\subsection{Cryptographic Primitives}
Here, we give definitions of basic cryptographic primitives. 

\begin{definition}[One-Way Functions (OWFs)]
\label{def:OWFs}
A function $f:\bit^*\to\bit^*$ is a (quantumly-secure) one-way function (OWF) if
it is computable in classical deterministic polynomial-time, and
for any QPT adversary $\cA$, there exists a negligible function $\negl$ such that
\begin{equation} \label{eq:OWF_condition}
\Pr[f(x')=f(x):
x\gets\bit^\secp,
x'\gets\cA(1^\secp,f(x))
]
\le\negl(\secp).
\end{equation} 
\end{definition}

\begin{definition}[Quantum Symmetric Key Encryption (QSKE)~\cite{TQC:MorYam24}]\label{def:QSKE} 
    A QSKE scheme 
    is a tuple $(\Gen,\Enc,\Dec)$ of QPT algorithms with the following syntax: 
    \begin{itemize}
        \item $\Gen(1^\secp)\to\sk$: A key generation algorithm takes the security parameter $1^\secp$ as input and outputs a classical secret key $\sk$.
        \item $\Enc(\sk,\msg)\to \ct$: An encryption algorithm takes a secret key $\sk$ and a message $\msg\in\bit^*$ as input and outputs a quantum ciphertext $\ct$.
        \item $\Dec(\sk,\ct)\to \msg'$: A decryption algorithm takes a secret key $\sk$ and a ciphertext $\ct$ as input and outputs a message $\msg'\in\bit^*$.
    \end{itemize}
    We require the following correctness and  IND-CPA security:
    \begin{itemize}
        \item \textbf{Correctness:} For all $\msg\in \bit^*$ of polynomial length in $\secp$, 
        \begin{align}
            \Pr \left[\msg'=\msg:
            \begin{gathered}
                \sk\gets\Gen(1^\secp) \\ 
                \ct\gets\Enc(\sk,\msg) \\
                \msg'\gets\Dec(\sk,\ct)
            \end{gathered}
            \right] \ge 1-\negl(\secp).
        \end{align}
        \item \textbf{IND-CPA Security:} For a security parameter $\secp\in\N$ and a bit $b\in\bit$, consider the following game between a challenger and an adversary $\cA$:
        \begin{enumerate}
            \item The challenger runs $\sk\gets\Gen(1^\secp)$.
            \item $\cA$ can make arbitrarily many classical queries to the encryption oracle, which takes a message $\msg\in \bit^*$ as input and returns $\Enc(\sk,\msg)$.
            \item $\cA$ chooses $(\msg_0,\msg_1)\in(\bit^*)^2$ of the same length and sends them to the challenger.
            \item The challenger runs $\ct_b\gets\Enc(\sk,\msg_b)$ and sends $\ct_b$ to $\cA$.
               \item Again, $\cA$ can make arbitrarily many classical queries to the encryption oracle. 
            \item $\cA$ outputs $b'$.
        \end{enumerate}
        We say that a QSKE scheme satisfies the IND-CPA security if for any QPT adversary $\cA$,
        \begin{align}
            | \Pr[b'=1|b=1]-\Pr[b'=1|b=0] | \le\negl(\secp).
        \end{align}
    \end{itemize}
\end{definition}

We define SKE and PKE schemes in the quantum-computation classical-communication (QCCC) model, in which all local computations are quantum and all communication is classical.

\begin{definition}[QCCC SKE~\cite{STOC:KhuTom24}]\label{def:QCCC_SKE}
A QCCC SKE scheme is defined similarly to a QSKE scheme as defined in \Cref{def:QSKE} except that a ciphertext $\ct$ output by $\Enc$ is required to be classical. 
\end{definition}

\begin{definition}[QCCC Public Key Encryption (QCCC PKE)~\cite{STOC:KhuTom24}]\label{def:QCCC_PKE}
    A QCCC PKE scheme 
    is a tuple $(\Gen,\Enc,\Dec)$ of QPT algorithms with the following syntax:
    \begin{itemize}
        \item $\Gen(1^\secp)\to(\pk,\sk)$: A key generation algorithm takes the security parameter $1^\secp$ as input and outputs a 
        classical public key $\pk$ and 
        classical secret key $\sk$.
        \item $\Enc(\pk,\msg)\to \ct$: An encryption algorithm takes a public key $\pk$ and a message $\msg\in\bit^*$ as input and outputs a classical ciphertext $\ct$.
        \item $\Dec(\sk,\ct)\to \msg'$: A decryption algorithm takes a secret key $\sk$ and a ciphertext $\ct$ as input and outputs a message $\msg'\in\bit^*$.
    \end{itemize}
    We require the following correctness and  IND-CPA security:
    \begin{itemize}
        \item \textbf{Correctness:}  For all $\msg\in \bit^*$ of polynomial length in $\secp$, 
        \begin{align}
            \Pr \left[\msg'=\msg:
            \begin{gathered}
                (\pk,\sk)\gets\Gen(1^\secp) \\ 
                \ct\gets\Enc(\pk,\msg) \\
                \msg'\gets\Dec(\sk,\ct)
            \end{gathered}
            \right] \ge 1-\negl(\secp).
        \end{align}
        \item \textbf{IND-CPA Security:} For a security parameter $\secp\in\N$ and a bit $b\in\bit$, consider the following game between a challenger and an adversary $\cA$:
        \begin{enumerate}
            \item The challenger runs $(\pk,\sk)\gets\Gen(1^\secp)$ and sends $\pk$ to $\cA$. 
            \item $\cA$ chooses $\msg_0,\msg_1\in(\bit^*)^2$ of the same length and sends them to the challenger.
            \item The challenger runs $\ct_b\gets\Enc(\sk,\msg_b)$ and sends $\ct_b$ to $\cA$.
            \item $\cA$ outputs $b'$.
        \end{enumerate}
        We say that a QCCC PKE scheme satisfies the IND-CPA security if for any QPT adversary $\cA$,
        \begin{align}
            | \Pr[b'=1|b=1]-\Pr[b'=1|b=0] | \le\negl(\secp).
        \end{align}
    \end{itemize}
\end{definition}

We omit definitions of other cryptographic primitives, such as EV-OWPuzzs~\cite{C:ChuGolGra24}, OWPuzzs~\cite{STOC:KhuTom24}, OWSGs~\cite{TQC:MorYam24}, QEFID pairs~\cite{C:ChuGolGra24}, PRUs~\cite{C:JiLiuSon18}, and EFI pairs~\cite{ITCS:BCQ23}, since we obtain them only as corollaries. For their definitions, we refer the reader to the respective cited works. 
One remark regarding the definition of OWSGs is that, unless stated otherwise, we refer to OWSGs with mixed-state outputs as defined in~\cite{TQC:MorYam24}. 

\subsection{Complexity Theory}
Here we explain basic complexity classes we use.
\begin{definition}[i.o.BQP]
    A promise problem $\Pi=(\Pi_{yes},\Pi_{no})$ is in $\mathsf{i.o.BQP}$ if there exist 
    a QPT algorithm $Q$ and infinitely many $\secp\in\N$ such that
    for all $x\in \Pi_{yes}\cup\Pi_{no}$,
    \begin{itemize}
        \item if $x\in\Pi_{yes}\cap\bit^\secp$, then $\Pr[1\gets Q(x)]\ge 2/3$.
        \item if $x\in\Pi_{no}\cap\bit^\secp$, then $\Pr[1\gets Q(x)]\le 1/3$.
    \end{itemize}
\end{definition}

\begin{definition}[UP]
   A promise problem $\Pi=(\Pi_{yes},\Pi_{no})$ is in $\mathsf{UP}$ if there exist a classical polynomial-time (deterministic) Turing machine $M$ and a polynomial $m$ such that
    \begin{itemize}
        \item if $x\in\Pi_{yes}$, then there exists a unique $w\in\bit^{m(|x|)}$ such that $M(x,w)=1$.
        \item if $x\in\Pi_{no}$, then for all $w\in\bit^{m(|x|)}$, $M(x,w)=0$.
    \end{itemize}
\end{definition}

The following lemma follows from the Valiant-Vazirani theorem \cite{VV86}.
\begin{lemma}\label{cor:VV}
    If $\mathsf{NP}\nsubseteq\mathsf{i.o.BQP}$, then $\mathsf{UP}\nsubseteq\mathsf{i.o.BQP}$.
\end{lemma}

\section{Definitions of Quantum Obfuscation}
We introduce definitions of quantum indistinguishability obfuscation for classical circuits and its variants. 
\begin{definition}[Quantum iO for Classical Circuits]\label{def:qiO}
   A quantum indistinguishability obfuscator (quantum iO) for classical circuits consists of two QPT algorithms $(\Obf,\Eval)$ with the following syntax:
    \begin{itemize}
        \item $\Obf(1^\secp,C)\to\obfC$: An obfuscation algorithm takes the security parameter $1^\secp$ and a classical circuit $C$ as input and outputs a quantum state $\obfC$, which we refer to as an obfuscated encoding of $C$.  
        \item $\Eval(\obfC,x)\to y$: An evaluation algorithm takes an obfuscated encoding $\obfC$ and a classical input $x$ as input and outputs a classical output $y$.
    \end{itemize}
 We require the following correctness and security.
    \begin{itemize}
        \item \textbf{Correctness:} For any family 
        $\{C_\secp\}_{\secp\in \mathbb{N}}$ of polynomial-size classical circuits of input length $n_\secp$, and for any polynomial $p$, there exists $N\in\N$ such that 
        \begin{align}
            \Pr_{\obfC_\secp\gets \Obf(1^\secp,C_\secp)}\left[
            \forall x\in \bit^{n_\secp}, 
            \Pr[\Eval(\obfC_\secp,x)=C_\secp(x)]\ge 1-\frac{1}{p(\secp)}\right] \ge 1-\frac{1}{p(\secp)}
        \end{align}
        holds for all $\secp\ge N$,
        where the inner probability is taken over the randomness of the execution of $\Eval(\obfC_\secp,x)$. 
        \item \textbf{Security:} 
        For any families 
        $\{C_{0,\secp}\}_{\secp\in \mathbb{N}}$ and
         $\{C_{1,\secp}\}_{\secp\in \mathbb{N}}$
        of polynomial-size classical circuits such that $C_{0,\secp}$ and $C_{1,\secp}$ are functionally equivalent and of the same size, and 
        for any non-uniform QPT adversary $\cA$,   
        \begin{align}
        | \Pr[1\gets\cA(1^\secp,\Obf(1^\secp,C_{0,\secp}))] - \Pr[1\gets\cA(1^\secp,\Obf(1^\secp,C_{1,\secp}))] | \le \negl(\secp).
    \end{align}
    \end{itemize}
\end{definition}
\begin{remark} 
The correctness notion defined above may seem strong, as it requires that the evaluation returns the correct output for all inputs simultaneously with overwhelming probability. However, this stronger guarantee can be generically achieved assuming only a quantum iO with a weaker, input-wise correctness—that is, for each fixed input, the evaluation returns the correct output with overwhelming probability. To achieve the stronger correctness, we can simply repeat the evaluation algorithm multiple times and take a majority vote. By repeating sufficiently many times, the failure probability for each input can be reduced exponentially, and a union bound then ensures that the overall probability of failing on any input remains negligible. This argument closely parallels the classical case presented in~\cite[Appendix B]{FOCS:KMNPRY14}, and thus we omit the details. 
\end{remark}
\begin{remark}
We require the security of iO to hold against non-uniform quantum QPT adversaries, even though our final goal is to construct cryptographic primitives with uniform security. This is because a uniform version of the security, where a uniform QPT adversary $\cA$ chooses two functionally equivalent circuits $C_0$ and $C_1$ and then tries to distinguish obfuscations of them, would not suffice for our purpose, since we must consider a reduction algorithm that hardwires an arbitrary choice of an $\mathsf{NP}$ instance into the circuits. 
Although this point is not explicitly discussed, we believe the same applies even in the classical setting of~\cite{FOCS:KMNPRY14}. 
In addition, we note that the non-uniform security notion aligns more closely with the standard formalization in the literature on iO. 
\end{remark}

\begin{definition}[Variations of Quantum iO]\label{def:variant_qiO}
For $(\mathtt{X},\mathtt{Y},\mathtt{Z})\in \{\QQ,\CC\}^3$, $(\mathtt{X},\mathtt{Y},\mathtt{Z})$-iO for classical circuits is defined similarly to quantum iO for classical circuits as defined in \Cref{def:qiO} except that:
\begin{itemize}
    \item If $\mathtt{X}=\CC$, $\Obf$ is a PPT algorithm whereas if $\mathtt{X}=\QQ$,  $\Obf$  is a QPT algorithm;
    \item If $\mathtt{Y}=\CC$, $\Eval$ is a PPT algorithm whereas if $\mathtt{Y}=\QQ$,  $\Eval$  is a QPT algorithm;
       \item If $\mathtt{Z}=\CC$, an encoding $\obfC$ output by $\Obf$ is a classical string whereas if $\mathtt{Z}=\QQ$,  $\obfC$  is a quantum state. 
\end{itemize}
\end{definition}
\begin{remark}
While there are $8$ possible choices for $(\mathtt{X},\mathtt{Y},\mathtt{Z})$, some of them are meaningless.  In particular, it makes sense to have $\mathtt{Z}=\QQ$ only if $\mathtt{X}=\mathtt{Y}=\QQ$ since classical $\Obf$ cannot output quantum $\obfC$ and classical $\Eval$ cannot take quantum $\obfC$ as input. Thus, there are $5$ meaningful choices: $(\QQ,\QQ,\QQ)$,  $(\QQ,\QQ,\CC)$, $(\QQ,\CC,\CC)$, $(\CC,\QQ,\CC)$, and $(\CC,\CC,\CC)$ (see \Cref{figure} for their relationship).
$(\QQ,\QQ,\QQ)$-iO corresponds to quantum iO as defined in \Cref{def:qiO} and $(\CC,\CC,\CC)$-iO corresponds to (post-quantum) classical iO. 
We stress that we consider obfuscation of \emph{classical} circuits and security against \emph{quantum} adversaries in all the variant.  
\end{remark}

In the security definition of iO, the two circuits are required to be functionally equivalent, that is, they must agree on all inputs. The notion of differing-inputs obfuscation (diO)~\cite{JACM:BGIRSVY12,EPRINT:ABGSZ13,TCC:BoyChuPas14} relaxes this requirement by allowing circuits that may differ on some inputs, as long as those inputs are hard to find. This results in a strictly stronger security notion than iO. It is known that in the classical setting, iO and diO are equivalent when the number of differing inputs is polynomial. We observe that this equivalence extends to the quantum setting as well. For our purposes, we present the definition of diO in the case where there is only a single differing input, which suffices for our applications.

\begin{definition}[Single-Point Differing-Inputs Obfuscation (diO)]\label{def:diO}  
For $(\mathtt{X},\mathtt{Y},\mathtt{Z})\in \{\QQ,\CC\}^3$, $(\mathtt{X},\mathtt{Y},\mathtt{Z})$-single-point diO for classical circuits is defined similarly to $(\mathtt{X},\mathtt{Y},\mathtt{Z})$-iO except that the security is replaced with extractability defined as follows: 
\begin{itemize}
     \item \textbf{Extractability (for single-differing-point):} 
     For any QPT adversary $\cA$ and any polynomial $p$, there exist a QPT algorithm $\Ext$ and a polynomial $q$ for which the following holds. 
    For any pair of families of polynomial-size classical circuits
        $\{C_{0,\secp}\}_{\secp\in \mathbb{N}}$ and
         $\{C_{1,\secp}\}_{\secp\in \mathbb{N}}$, such that for each $\secp$, $C_{0,\secp}$ and $C_{1,\secp}$ 
        have the same size and input length, and differ on at most a single input,       
        and for any family of polynomial-size classical strings $\{z_\secp\}_{\secp\in \mathbb{N}}$, 
        the following holds for all sufficiently large $\secp\in\mathbb{N}$: 
    \begin{align}
        & \Pr \left[ b'=b : 
        \begin{gathered}
            b\gets\bit \\
            \obfC_\secp\gets\Obf(1^\secp,C_{b,\secp}) \\
            b'\gets\cA(1^\secp,\obfC_\secp,C_{0,\secp},C_{1,\secp},z_\secp)
        \end{gathered}
        \right] \ge \frac{1}{2}+\frac{1}{p(\secp)} \\ 
        &\quad \Longrightarrow 
        \Pr\left[ C_{0,\secp}(x)\neq C_{1,\secp}(x) : x\gets\Ext(1^\secp,C_{0,\secp},C_{1,\secp},z_\secp) \right] \ge \frac{1}{q(\secp)}. 
    \end{align}
\end{itemize}
\end{definition}
\begin{remark}
Zhandry~\cite{C:Zhandry23} observed that defining diO involves subtle challenges when considering security against quantum adversaries with quantum advice. In contrast, we restrict our attention to quantum adversaries with classical advice. As a result, these complications do not arise in our setting, allowing us to define diO in a manner that closely mirrors the classical definition from~\cite{TCC:BoyChuPas14}. 
\end{remark}
In the classical advice setting considered above, the equivalence between iO and single-point diO can be proven using essentially the same argument as in the classical case, as shown in~\cite{TCC:BoyChuPas14}.
\begin{lemma}\label{lem:diO}
For $(\mathtt{X},\mathtt{Y},\mathtt{Z})\in \{\QQ,\CC\}^3$, if $(\Obf,\Eval)$ is an $(\mathtt{X},\mathtt{Y},\mathtt{Z})$-iO for classical circuits, then it is also $(\mathtt{X},\mathtt{Y},\mathtt{Z})$-single-point diO for classical circuits.
\end{lemma}

\section{Main Technical Theorem}
We prove a technical theorem that is the basis of all our cryptographic implications. 

Let $s_{min}$ be a polynomial such that, 
for any $m\in \mathbb{N}$ and any $k\in \bit^m$, 
there exist classical circuits of size at most $s_{min}(m)$ that compute the following functions on $m$-bit inputs:
\begin{itemize}
    \item the point function at target point $k$, which outputs $1$ on input $k$ and $0$ on all other inputs; and 
    \item the zero function, which outputs 0 on all $m$-bit inputs. 
\end{itemize} 
For 
$m\in \mathbb{N}$, 
$k\in \bit^m$, and $s\ge s_{min}(m)$, let $P_{k,s}$ denote a \emph{canonical} classical circuit of size $s$ that computes the point function on the target point $k$, and
let $Z_{m,s}$ be a \emph{canonical} classical circuit of size $s$ that computes the zero-function on $m$-bit inputs. 
Here, ``canonical'' refers to a fixed but arbitrary choice of circuit construction, provided that the descriptions of $P_{k,s}$ and $Z_{m,s}$ are computable in classical polynomial time from $(k,1^s)$ and $(1^m,1^s)$,  respectively. The specific choice of canonical circuits does not affect our results.


Then we prove the following theorem. 
\begin{theorem}\label{thm:main}
    Suppose $\mathsf{NP}\nsubseteq\mathsf{i.o.BQP}$ and 
    $(\Obf,\Eval)$ is an $(\mathtt{X},\mathtt{Y},\mathtt{Z})$-iO for classical circuits for $(\mathtt{X},\mathtt{Y},\mathtt{Z})\in \{\QQ,\CC\}^3$. 
Then, there are classical-polynomial-time-computable polynomials $m$ and $s$  such that 
    for any polynomial $\ell$, 
    the following two distributions (over classical bit strings if $\mathtt{Z}=\CC$ and over quantum states if $\mathtt{Z}=\QQ$) are 
    computationally indistinguishable against uniform QPT adversaries:
    \begin{itemize}
        \item $\cD_0(\secp)$: Sample $k\gets\bit^{m(\secp)}$, 
        run $\hat{P}_{k,s(\secp)}^i\gets \Obf(1^\secp,P_{k,s(\secp)})$ for $i\in [\ell(\secp)]$,  
        and output $(1^\secp,\hat{P}_{k,s(\secp)}^1,\hat{P}_{k,s(\secp)}^2,\ldots,\hat{P}_{k,s(\secp)}^{\ell(\secp)})$. 
         \item $\cD_1(\secp)$: Run $\hat{Z}_{m(\secp),s(\secp)}^i\gets \Obf(1^\secp,Z_{m(\secp),s(\secp)})$ for $i\in [\ell(\secp)]$ 
        and output $(1^\secp,\hat{Z}_{m(\secp),s(\secp)}^1,\hat{Z}_{m(\secp),s(\secp)}^2,\ldots,\hat{Z}_{m(\secp),s(\secp)}^{\ell(\secp)})$. 
        \end{itemize}
\end{theorem}

\begin{proof}[Proof of \cref{thm:main}]
    By \cref{cor:VV}, we have $\mathsf{UP}\nsubseteq\mathsf{i.o.BQP}$.
    Then, there exists a promise problem $\Pi=(\Pi_{yes},\Pi_{no})\in\mathsf{UP}$ such that $\Pi\notin\mathsf{i.o.BQP}$.
    By the definition of $\mathsf{UP}$, there exist a polynomial $m$ and a deterministic polynomial-time Turing machine $M$ such that
    \begin{itemize}
        \item If $x\in\Pi_{yes}$, then there exists a unique $w\in\bit^{m(|x|)}$ such that $M(x,w)=1$.
        \item If $x\in\Pi_{no}$, $M(x,w)=0$ for any $w\in\bit^{m(|x|)}$.
    \end{itemize}
    Without loss of generality, we can assume that $m$ is classical-polynomial-time-computable because we can pad $w$ so that its length matches (an upper bound of) the running time of $M$.

    For the sake of contradiction, let us assume that for any classical-polynomial-time-computable polynomial $s$, the following holds:
    There exist polynomials $\ell$ and $p$, and a uniform QPT algorithm $\cA$ such that
    \begin{align}\label{eq:assump_main}
        \Pr_{k\gets\bit^{m(\secp)}} [ 1\gets\cA(1^\secp,\hat{P}_{k,s(\secp)}^1,...,\hat{P}_{k,s(\secp)}^{\ell(\secp)}) ] - \Pr [1\gets\cA(1^\secp,\hat{Z}_{m(\secp),s(\secp)}^1,...,\hat{Z}_{m(\secp),s(\secp)}^{\ell(\secp)})] \ge \frac{1}{p(\secp)}
    \end{align}
    holds for infinitely many $\secp$, where $\hat{P}_{k,s(\secp)}^i\gets\Obf(1^\secp,P_{k,s(\secp)})$ and $\hat{Z}_{m(\secp),s(\secp)}^i\gets\Obf(1^\secp,Z_{m(\secp),s(\secp)})$ for each $i\in[\ell(\secp)]$.
    Our goal is to construct a QPT algorithm that solves $\Pi$ for infinitely many input lengths, thereby showing that $\Pi\in\mathsf{i.o.BQP}$.
    To do this, it suffices to show that there exist a QPT algorithm $\cB$ and a polynomial $r$ such that for infinitely many $\secp\in\mathbb{N}$ 
    \begin{align}\label{eq:find_w}
        \Pr[M(x,w)=1:w\gets\cB(x)] \ge \frac{1}{r(\secp)}
    \end{align}
    is satisfied for all $x\in\Pi_{yes}\cap\bit^\secp$.
    Then, a QPT algorithm $\cC$ that on input $x\in\bit^\secp$, runs $w\gets\cB(x)$ and outputs $M(x,w)$ satisfies,
    \begin{itemize}
        \item if $x\in\Pi_{yes}$, $\Pr[1\gets\cC(x)]=\Pr[M(x,w)=1:w\gets\cB(x)]\ge\frac{1}{r(\secp)}$,
        \item if $x\in\Pi_{no}$, $\Pr[1\gets\cC(x)]=\Pr[M(x,w)=1:w\gets\cB(x)]=0$,
    \end{itemize}
    for infinitely many $\secp\in\mathbb{N}$.
    The completeness-soundness gap is $1/r(\secp)=1/\poly(\secp)$ and therefore $\Pi\in\mathsf{i.o.BQP}$. 

    In the remaining part, we show the existence of a polynomial $r$ and a QPT algorithm $\cB$ that satisfy \cref{eq:find_w}.
    There exists a family $\{V_\secp[x,v]\}_{\secp\in\N}$ of polynomial-size classical circuits such that for each $\secp$, $V_\secp[x,v]$ is parametrized by $x\in\bit^\secp$ and $v\in\bit^{m(\secp)}$, operates on $m(\secp)$ bits, and computes the function
    \begin{align}
        V_\secp[x,v](z) := 
        \begin{cases}
            1 & \text{if } M(x,v\oplus z)=1 \\ 
            0 & \text{otherwise},
        \end{cases}
    \end{align}
    where the description of $V_\secp[x,v]$ is computable in classical polynomial time from $(x,v)$.
    Let $s(\secp)$ be the size of $V_\secp[x,v]$. Then, $s$ is classical-polynomial-time-computable because there exists a Turing machine that on input $1^\secp$, computes $1^{m(\secp)}$, computes the description of $V_\secp[1^\secp,1^{m(\secp)}]$ from $(1^\secp,1^{m(\secp)})$, and outputs $s(\secp)=|V_\secp[1^\secp,1^{m(\secp)}]|$.
    We can choose $s$ such that $s(\secp)\ge s_{min}(m(\secp))$ by adding dummy gates that do not change the functionality of $V_\secp[x,v]$. 

    For each $x\in\Pi_{yes}$ and $v\in\bit^{m(|x|)}$, $V_{|x|}[x,v]$ has the same functionality with $P_{w\oplus v,s(|x|)}$, where $w\in\bit^{m(|x|)}$ is the unique witness for $x$.
    
    Then, by the security of iO, for any uniform QPT algorithm $\cA'$ and for any polynomial $b$, there exists $N\in\N$ such that
    \begin{align}\label{eq:ind_V_P}
        &\Bigg| \Pr [ 1\gets\cA'(1^\secp,\hat{V}_{x,v}^1,...,\hat{V}_{x,v}^{\ell(\secp)}) ] - \Pr [1\gets\cA'(1^\secp,\hat{P}_{w\oplus v,s(\secp)}^1,...,\hat{P}_{w\oplus v,s(\secp)}^{\ell(\secp)}) ] \Bigg| \le \frac{1}{b(\secp)}
    \end{align}
    for all $v\in\bit^{m(\secp)}$, all $x\in\Pi_{yes}\cap\bit^\secp$ and all $\secp$ such that $\secp\ge N$,
    where $\hat{V}_{x,v}^i\gets\Obf(1^\secp,V_\secp[x,v])$ and $\hat{P}_{w\oplus v,s(\secp)}^i\gets\Obf(1^\secp,P_{w\oplus v,s(\secp)})$ for each $i\in[\ell(\secp)]$.
    Then, for the QPT algorithm $\cA$ that satisfies \cref{eq:assump_main}, there exist infinitely many $\secp\in\N$ such that for all $x\in\Pi_{yes}\cap\bit^\secp$, 
    \begin{align}
        &\Pr_{v\gets\bit^{m(\secp)}}[ 1\gets\cA(1^\secp,\hat{V}_{x,v}^1,...,\hat{V}_{x,v}^{\ell(\secp)})] - \Pr[ 1\gets\cA(1^\secp,\hat{Z}_{m(\secp),s(\secp)}^1,...,\hat{Z}_{m(\secp),s(\secp)}^{\ell(\secp)}) ] \\  
        &= \Pr_{v\gets\bit^{m(\secp)}}[ 1\gets\cA(1^\secp,\hat{P}_{v,s(\secp)}^1,...,\hat{P}_{v,s(\secp)}^{\ell(\secp)}) ] - \Pr[1\gets\cA(1^\secp,\hat{Z}_{m(\secp),s(\secp)}^1,...,\hat{Z}_{m(\secp),s(\secp)}^{\ell(\secp)})] \\ 
        &\quad + \Pr_{v\gets\bit^{m(\secp)}}[ 1\gets\cA(1^\secp,\hat{V}_{x,v}^1,...,\hat{V}_{x,v}^{\ell(\secp)}) ] - \Pr_{v\gets\bit^{m(\secp)}}[1\gets\cA(1^\secp,\hat{P}_{v,s(\secp)}^1,...,\hat{P}_{v,s(\secp)}^{\ell(\secp)})] \\ 
        &= \Pr_{v\gets\bit^{m(\secp)}}[ 1\gets\cA(1^\secp,\hat{P}_{v,s(\secp)}^1,...,\hat{P}_{v,s(\secp)}^{\ell(\secp)}) ] - \Pr [1\gets\cA(1^\secp,\hat{Z}_{m(\secp),s(\secp)}^1,...,\hat{Z}_{m(\secp),s(\secp)}^{\ell(\secp)})] \\ 
        &\quad + \Pr_{v\gets\bit^{m(\secp)}}[ 1\gets\cA(1^\secp,\hat{V}_{x,v}^1,...,\hat{V}_{x,v}^{\ell(\secp)}) ] - \Pr_{v\gets\bit^{m(\secp)}}[1\gets\cA(1^\secp,\hat{P}_{w\oplus v,s(\secp)}^1,...,\hat{P}_{w\oplus v,s(\secp)}^{\ell(\secp)})] \\ 
        &\ge \frac{1}{p(\secp)} - \frac{1}{2p(\secp)} \quad (\text{By \cref{eq:assump_main,eq:ind_V_P}.}) \\ 
        &= \frac{1}{2p(\secp)}, \label{eq:dist_iO}
    \end{align}
    where $\hat{V}_{x,v}^i\gets\Obf(1^\secp,V_\secp[x,v])$, $\hat{P}_{k,s(\secp)}^i\gets\Obf(1^\secp,P_{k,s(\secp)})$, and $\hat{Z}_{m(\secp),s(\secp)}^i\gets\Obf(1^\secp,Z_{m(\secp),s(\secp)})$ for each $i\in[\ell(\secp)]$.
    Let us consider the following QPT algorithm $\cE$:
    \begin{enumerate}
        \item Take $(1^\secp,\hat{C}^1,...,\hat{C}^{\ell(\secp)},V_\secp[x,v],Z_{m(\secp),s(\secp)})$ as input, where $\hat{C}^i\in\{\hat{V}_{x,v}^i,\hat{Z}^i_{m(\secp),s(\secp)}\}$ for all $i\in[\ell(\secp)]$.
        \item Run $b\gets\cA(1^\secp,\hat{C}^1,...,\hat{C}^{\ell(\secp)})$.
        \item Output $b$.
    \end{enumerate}
    By \cref{eq:dist_iO}, there exist infinitely many $\secp\in\mathbb{N}$ such that for all $x\in\Pi_{yes}\cap\bit^\secp$, 
    \begin{align}
        &\underset{v\gets\bit^{m(\secp)}}{\mathbb{E}} \Bigg[ \Pr[1\gets\cE(1^\secp,\hat{V}_{x,v}^1,...,\hat{V}_{x,v}^{\ell(\secp)},V_\secp[x,v],Z_{m(\secp),s(\secp)})] \\ 
        &\qquad - \Pr[1\gets\cE(1^\secp,\hat{Z}_{m(\secp),s(\secp)}^1,...,\hat{Z}_{m(\secp),s(\secp)}^{\ell(\secp)},V_\secp[x,v],Z_{m(\secp),s(\secp)})] \Bigg] \\ 
        &= \Pr_{v\gets\bit^{m(\secp)}}[1\gets\cA(1^\secp,\hat{V}_{x,v}^1,...,\hat{V}_{x,v}^{\ell(\secp)})] - \Pr[1\gets\cA(1^\secp,\hat{Z}_{m(\secp),s(\secp)}^1,...,\hat{Z}_{m(\secp),s(\secp)}^{\ell(\secp)})] \\
        &\ge \frac{1}{2p(\secp)}. \label{eq:dist_iO_2}
    \end{align}
    For each $\secp\in\mathbb{N}$ and each $x\in\bit^\secp$, define a set 
    \begin{align}
        \Good_{\secp,x} := \Bigg\{ v\in\bit^{m(\secp)} : 
        \begin{aligned}
            &\Pr[ 1\gets\cE(1^\secp,\hat{V}_{x,v}^1,...,\hat{V}_{x,v}^{\ell(\secp)},V_\secp[x,v],Z_{m(\secp),s(\secp)}) ] \\ 
            &- \Pr [1\gets\cE(1^\secp,\hat{Z}_{m(\secp),s(\secp)}^1,...,\hat{Z}_{m(\secp),s(\secp)}^{\ell(\secp)},V_\secp[x,v],Z_{m(\secp),s(\secp)})] \ge \frac{1}{4p(\secp)}
        \end{aligned} \Bigg\}.
    \end{align}
    By \cref{eq:dist_iO_2}, 
    \begin{align}
        \frac{1}{2p(\secp)} 
        &\le \Pr_{v\gets\bit^{m(\secp)}}[v\in\Good_{\secp,x}] + \left(1-\Pr_{v\gets\bit^{m(\secp)}}[v\in\Good_{\secp,x}]\right) \frac{1}{4p(\secp)}.
    \end{align}
    Thus, for infinitely many $\secp\in\mathbb{N}$ and for all $x\in\Pi_{yes}\cap\bit^\secp$,
    \begin{align}
        \Pr_{v\gets\bit^{m(\secp)}}[v\in\Good_{\secp,x}] \ge \frac{1}{4p(\secp)-1}.
    \end{align}
    By \cref{lem:diO}, $(\Obf,\Eval)$ is also a single-point diO for classical circuits.
    Therefore, there exist a QPT algorithm $\Ext$ and a polynomial $q$ such that the following holds: There exist infinitely many $\secp\in\mathbb{N}$ such that for all $x\in\Pi_{yes}\cap\bit^\secp$ and all $v\in\Good_{\secp,x}$,
    \begin{align}
        \Pr\left[V_\secp[x,v](z)\neq Z_{m(\secp),s(\secp)}(z):
        \begin{gathered}
            z\gets\Ext(1^\secp,V_\secp[x,v],Z_{m(\secp),s(\secp)})
        \end{gathered}\right]
        \ge \frac{1}{q(\secp)}.
    \end{align}
    By using such $\Ext$, we construct a QPT algorithm $\cB$ as follows:
    \begin{itemize}
        \item Take $x\in\bit^*$ as input. Set $\secp\coloneqq|x|$.
        \item Compute $m(\secp)$. Here $m(\secp)$ is the classical-polynomial-time-computable polynomial that corresponds to the witness length for $x$. 
        \item Sample $v\gets\bit^{m(\secp)}$.
        \item Compute $s(\secp)$. Here $s(\secp)$ is the classical-polynomial-time-computable polynomial that corresponds to the size of $V_\secp[x,v]$.  
        \item Run $z\gets\Ext(1^\secp,V_\secp[x,v],Z_{m(\secp),s(\secp)})$.
        \item Output $z\oplus v$.
    \end{itemize}
    Then, 
    for infinitely many $\secp\in\mathbb{N}$,
    \begin{align}
        &\Pr[M(x,w)=1:w\gets\cB(x)] \\
        &\ge \Pr\left[v\in\Good_{\secp,x} \land V_\secp[x,v](z)\neq Z_{m(\secp),s(\secp)}(z):
        \begin{gathered}
            v\gets\bit^{m(\secp)}; \\
            z\gets\Ext(1^\secp,V_\secp[x,v],Z_{m(\secp),s(\secp)})
        \end{gathered}
        \right] \\ 
        &\ge \frac{1}{q(\secp)(4p(\secp)-1)},
    \end{align}
    holds for all $x\in \Pi_{yes}\cap\bit^\secp$.
    We have the QPT algorithm $\cB$ and a polynomial $r(\secp):=q(\secp)(4p(\secp)-1)$ that satisfy \cref{eq:find_w}, and therefore we complete the proof.
\end{proof}

\section{Cryptographic Implications}
In this section, we use \Cref{thm:main} to demonstrate cryptographic implications of various quantum iO variants (as defined in \Cref{def:variant_qiO}), when combined with the worst-case hardness of $\mathsf{NP}$. 
\subsection{Q-Obf, Q-Eval, and Q-Encoding}
Here, we study implications of $(\QQ,\QQ,\QQ)$-iO, where both $\Obf$ and $\Eval$ are quantum algorithms and an obfuscated encoding $\hat{C}$ is a quantum state.  
We prove the following theorem:
\begin{theorem}\label{thm:QQQ_QSKE}
    Suppose $\mathsf{NP}\nsubseteq\mathsf{i.o.BQP}$ and there exists $(\QQ,\QQ,\QQ)$-iO for classical circuits. Then there exists an IND-CPA secure QSKE scheme.
\end{theorem}

\begin{proof}[Proof of \cref{thm:QQQ_QSKE}]
    By \cref{thm:main}, there exist classical-polynomial-time-computable polynomials $m$ and $s$ such that for any polynomial $\ell$, the following two distributions over quantum states are computationally indistinguishable:
    \begin{itemize}
        \item $\cD_0(\secp)$: Sample $k\gets\bit^{m(\secp)}$, run $\hat{P}_{k,s(\secp)}^i\gets\Obf(1^\secp,P_{k,s(\secp)})$ for $i\in[\ell(\secp)]$, and output $(1^\secp,\hat{P}_{k,s(\secp)}^1,...,\hat{P}_{k,s(\secp)}^{\ell(\secp)})$.
        \item $\cD_1(\secp)$: Run $\hat{Z}_{m(\secp),s(\secp)}^i\gets\Obf(1^\secp,Z_{m(\secp),s(\secp)})$ for $i\in[\ell(\secp)]$ and output $(1^\secp,\hat{Z}_{m(\secp),s(\secp)}^1,...,\hat{Z}_{m(\secp),s(\secp)}^{\ell(\secp)})$.
    \end{itemize}
    Without loss of generality, it suffices to construct an IND-CPA secure QSKE scheme $(\Gen,\Enc,\Dec)$ for single-bit message.
    We construct $(\Gen,\Enc,\Dec)$ as follows:
    \begin{itemize}
        \item $\Gen(1^\secp)\to\sk$: Take the security parameter $1^\secp$ as input, compute $m(\secp)$, and sample $k\gets\bit^{m(\secp)}$. Output $\sk:=k$.
        \item $\Enc(\sk,b)\to\ct$: Take the secret key $\sk$ and a message $b\in\bit$ as input. Let $C_0:=Z_{m(\secp),s(\secp)}$ and $C_1:=P_{\sk,s(\secp)}$. Sample $\hat{C}_b\gets\Obf(1^\secp,C_b)$ and output $\ct:=\hat{C}_b$.
        \item $\Dec(\sk,\ct)\to b'$: Take $\sk$ and $\ct$ as input. Run $b'\gets\Eval(\ct,\sk)$ and output $b'$.
    \end{itemize}
    The correctness of $(\Gen,\Enc,\Dec)$ follows from the correctness of iO:
    For all $b\in\bit$,
    \begin{align}
        \Pr\left[b'=b:
        \begin{gathered}
            \sk\gets\Gen(1^\secp); \\
            \ct\gets\Enc(\sk,b); \\
            b'\gets\Dec(\sk,\ct)
        \end{gathered}\right] 
        &= \Pr\left[b'=b:
        \begin{gathered}
            k\gets\bit^{m(\secp)}; \\
            \hat{C}_b\gets\Obf(1^\secp,C_b); \\
            b'\gets\Eval(\hat{C}_b,k)
        \end{gathered}\right] \\ 
        &= \Pr\left[\Eval(\hat{C}_b,k)=C_b(k):
        \begin{gathered}
            k\gets\bit^{m(\secp)}; \\
            \hat{C}_b\gets\Obf(1^\secp,C_b)
        \end{gathered}\right] \\ 
        &\ge 1-\negl(\secp).
    \end{align}
    We show the IND-CPA security of $(\Gen,\Enc,\Dec)$.
    In the IND-CPA security game of $(\Gen,\Enc,\Dec)$, it suffices to consider any QPT adversary that makes polynomially many queries only on message 1 to the encryption oracle because the ciphertext for message 0 can be computed without the secret key by simply running $\Obf(1^\secp,Z_{m(\secp),s(\secp)})$.
    Thus, our goal is to show that for any polynomial $t$, the following two distributions are computationally indistinguishable:
    \begin{itemize}
        \item $\cD_0'(\secp)$: Sample $\sk\gets\Gen(1^\secp)$, run $\ct_1^i\gets\Enc(\sk,1)$ for $i\in[t(\secp)]$ and $\ct^*_1\gets\Enc(\sk,1)$, and output $(\ct_1^1,...,\ct_1^{t(\secp)},\ct^*_1)$. 
        \item $\cD_1'(\secp)$: Sample $\sk\gets\Gen(1^\secp)$ run $\ct_1^i\gets\Enc(\sk,1)$ for $i\in[t(\secp)]$ and $\ct^*_0\gets\Enc(\sk,0)$, and output $(\ct_1^1,...,\ct_1^{t(\secp)},\ct^*_0)$.
    \end{itemize}
    When $\ell(\secp)=t(\secp)+1$, the distribution $\cD'_0(\secp)$ is identical to the distribution $\cD_0(\secp)$.
    Thus, for any polynomial $t$ and for any QPT adversary $\cA$,
    \begin{align}
        &\left| \Pr[1\gets\cA(1^\secp,\cD'_0(\secp))] - \Pr[1\gets\cA(1^\secp,\cD'_1(\secp))] \right| \\ 
        &= \left| \Pr[1\gets\cA(1^\secp,\cD_0(\secp))] - \Pr[1\gets\cA(1^\secp,\cD'_1(\secp))] \right| \\
        &\le \left| \Pr[1\gets\cA(1^\secp,\cD_0(\secp))] - \Pr[1\gets\cA(1^\secp,\cD_1(\secp))] \right| + \left| \Pr[1\gets\cA(1^\secp,\cD_1(\secp))] - \Pr[1\gets\cA(1^\secp,\cD'_1(\secp))] \right| \\ 
        &\le \negl(\secp) + \left| \Pr[1\gets\cA(1^\secp,\cD_1(\secp))] - \Pr[1\gets\cA(1^\secp,\cD'_1(\secp))] \right|. \quad (\text{By \cref{thm:main}.})
    \end{align}
    To show the computational indistinguishability between $\cD_0'(\secp)$ and $\cD'_1(\secp)$, it suffices to show that $\cD_1(\secp)$ and $\cD'_1(\secp)$ are computationally indistinguishable.
    For the sake of contradiction, assume that there exist polynomials $t$ and $p$ and a QPT algorithm $\cA$ such that
    \begin{align}
        \frac{1}{p(\secp)} 
        &\le \left| \Pr[1\gets\cA(1^\secp,\cD_1(\secp))] - \Pr[1\gets\cA(1^\secp,\cD'_1(\secp))] \right| \\
        &= \Bigg| \Pr[1\gets\cA(1^\secp,\hat{Z}^1_{m(\secp),s(\secp)},...,\hat{Z}^{t(\secp)}_{m(\secp),s(\secp)},\hat{Z}^{t(\secp)+1}_{m(\secp),s(\secp)})] \\ 
        &\qquad - \Pr_{k\gets\bit^{m(\secp)}}[1\gets\cA(1^\secp,\hat{P}^1_{k,s(\secp)},...,\hat{P}^{t(\secp)}_{k,s(\secp)},\hat{Z}^{t(\secp)+1}_{m(\secp),s(\secp)})] \Bigg|
    \end{align}
    for infinitely many $\secp\in\mathbb{N}$, where $\hat{Z}^i_{m(\secp),s(\secp)}\gets\Obf(1^\secp,Z_{m(\secp),s(\secp)})$ and $\hat{P}^i_{k,s(\secp)}\gets\Obf(1^\secp,P_{k,s(\secp)})$ for $i\in[t(\secp)+1]$.
    Let us consider a QPT algorithm $\cB$ that on input $(1^\secp,\hat{C}_b^1,...,\hat{C}_b^{t(\secp)})$, runs $\hat{Z}_{m(\secp),s(\secp)}\gets\Obf(1^\secp,Z_{m(\secp),s(\secp)})$ and $b'\gets\cA(1^\secp,\hat{C}_b^1,...,\hat{C}_b^{t(\secp)},\hat{Z}_{m(\secp),s(\secp)})$, and outputs $b'$.
    Then,
    \begin{align}
        &\Bigg| \Pr[1\gets\cB(1^\secp,\hat{Z}_{m(\secp),s(\secp)}^1,...,\hat{Z}_{m(\secp),s(\secp)}^{t(\secp)})] - \Pr_{k\gets\bit^{m(\secp)}}[1\gets\cB(1^\secp,\hat{P}^1_{k,s(\secp)},...,\hat{P}^{t(\secp)}_{k,s(\secp)})] \Bigg| \\ 
        &=
        \Bigg| \Pr[1\gets\cA(1^\secp,\hat{Z}^1_{m(\secp),s(\secp)},...,\hat{Z}^{t(\secp)}_{m(\secp),s(\secp)},\hat{Z}^{t(\secp)+1}_{m(\secp),s(\secp)})] \\ 
        &\qquad - \Pr_{k\gets\bit^{m(\secp)}}[1\gets\cA(1^\secp,\hat{P}^1_{k,s(\secp)},...,\hat{P}^{t(\secp)}_{k,s(\secp)},\hat{Z}^{t(\secp)+1}_{m(\secp),s(\secp)})] \Bigg| \\ 
        &\ge \frac{1}{p(\secp)}
    \end{align}
    for infinitely many $\secp\in\mathbb{N}$.
    This contradicts \cref{thm:main} and therefore the distributions $\cD'_1(\secp)$ and $\cD_1(\secp)$ are computationally indistinguishable.
    Hence we complete the proof of IND-CPA security.
\end{proof}

Since IND-CPA secure QSKE implies OWSGs and EFI pairs~\cite{TQC:MorYam24,STOC:KhuTom24,FOCS:BatJai24}, we have the following corollary.
\begin{corollary}\label{cor:QQQ_OWSG_EFI}
    Suppose $\mathsf{NP}\nsubseteq\mathsf{i.o.BQP}$ and there exists $(\QQ,\QQ,\QQ)$-iO for classical circuits. Then there exist OWSGs and EFI pairs. 
\end{corollary}

\subsection{Q-Obf, Q-Eval, and C-Encoding}
Here, we study implications of $(\QQ,\QQ,\CC)$-iO, where both $\Obf$ and $\Eval$ are quantum algorithms, but an obfuscated encoding $\hat{C}$ is classical. This can be regarded as an iO in the QCCC model. 
We prove the following theorem:
\begin{theorem}\label{thm:QQC_QCCC_SKE}
    Suppose $\mathsf{NP}\nsubseteq\mathsf{i.o.BQP}$ and there exists $(\QQ,\QQ,\CC)$-iO for classical circuits. Then there exists an IND-CPA secure QCCC SKE scheme. 
\end{theorem}

\begin{proof}[Proof of \cref{thm:QQC_QCCC_SKE}]
    The proof of this theorem is quite similar to \cref{thm:QQQ_QSKE}.
    The only difference is that the output of $\Obf$ is a classical string rather than a quantum state.
    The proof of \cref{thm:QQQ_QSKE} heavily relies on \cref{thm:main} that is also valid for $(\QQ,\QQ,\CC)$-iO.
    Thus, we can easily extend the proof of \cref{thm:QQQ_QSKE} to the $(\QQ,\QQ,\CC)$-iO.
\end{proof}

Since IND-CPA secure QCCC SKE implies EV-OWPuzz, which in turn implies OWPuzz, OWSGs, and QEFID pairs, and EFI pairs~\cite{STOC:KhuTom24,C:ChuGolGra24}, we have the following corollary. 
\begin{corollary}\label{cor:QQC_EVOWPuzz}
    Suppose $\mathsf{NP}\nsubseteq\mathsf{i.o.BQP}$ and there exists $(\QQ,\QQ,\CC)$-iO for classical circuits. Then there exist EV-OWPuzz, OWPuzz, OWSGs, QEFID pairs, and EFI pairs. 
\end{corollary}

\subsection{Q-Obf, C-Eval, and C-Encoding}\label{sec:QCC}
Here, we study implications of $(\QQ,\CC,\CC)$-iO, where $\Obf$ is a quantum algorithm, $\Eval$ is a classical algorithm, and an obfuscated encoding $\hat{C}$ is classical.  
We prove the following theorem:
\begin{theorem}\label{thm:QCC_QCCC_PKE}
    Suppose $\mathsf{NP}\nsubseteq\mathsf{i.o.BQP}$ and there exists $(\QQ,\CC,\CC)$-iO for classical circuits. Then there exists an IND-CPA secure QCCC PKE scheme. 
\end{theorem}
To prove this theorem, we first show that $(\QQ,\CC,\CC)$-iO can be modified to satisfy a stronger notion of correctness that holds even for fixed randomness.
\begin{lemma}\label{lem:fixed_randomness_correctness}
Suppose that there exists $(\QQ,\CC,\CC)$-iO for classical circuits. Then there exists $(\QQ,\CC,\CC)$-iO for classical circuits that satisfies the fixed 
randomness correctness, defined below:
\begin{itemize}
   \item \textbf{Fixed randomness correctness:} For any family 
        $\{C_\secp\}_{\secp\in \mathbb{N}}$ of polynomial-size classical circuits of input length $n_\secp$, 
        \begin{align}\label{eq:fixed_randomness_correctness}
            \Pr_{\substack{\obfC_\secp\gets \Obf(1^\secp,C_\secp)\\ r\gets \mathcal{R}_\secp}}\left[
            \forall x\in \bit^{n_\secp}, 
           \Eval(\obfC_\secp,x;r)=C_\secp(x)\right] \ge 1-\negl(\secp)
        \end{align}
        where $\mathcal{R}_\secp$
        denotes the randomness space of 
        $\Eval(\obfC_\secp,x)$, and
        $\Eval(\obfC_\secp,x;r)$ denotes the execution with the fixed randomness $r$.\footnote{We assume without loss of generality that the randomness space of $\Eval(\obfC_\secp,x)$ only depends on $\secp$ and does not depend on $C_\secp$ and $x$.}  
\end{itemize}
\end{lemma}
\begin{proof}[Proof of \Cref{lem:fixed_randomness_correctness}]
    The correctness (as in \Cref{def:qiO}) implies that, for any family 
        $\{C_\secp\}_{\secp\in \mathbb{N}}$ of polynomial-size classical circuits of input length $n_\secp$ and for any $x\in \bit^{n_\secp}$, 
                \begin{align}
            \Pr_{\substack{\obfC_\secp\gets \Obf(1^\secp,C_\secp)\\ r\gets \mathcal{R}_\secp}}\left[
           \Eval(\obfC_\secp,x;r)=C_\secp(x)\right] \ge 1-\negl(\secp) \ge 2/3.
        \end{align}
    Thus, if we modify $\Obf$ to run $M=O(\secp+n_\secp)$ times to output $M$ independently generated obfuscated encodings of $C_\secp$, and   $\Eval$ to evaluate each of the $M$ obfuscated encodings and take the majority result, we can ensure that it satisfies\footnote{Due to the modifications to $\Eval$, $\mathcal{R}_\secp$ is also updated accordingly; it now consists of $M$-tuples of the randomness used in the original $\Eval$.} 
      \begin{align}
            \Pr_{\substack{\obfC_\secp\gets \Obf(1^\secp,C_\secp)\\ r\gets \mathcal{R}_\secp}}\left[
           \Eval(\obfC_\secp,x;r)=C_\secp(x)\right]  \ge 1-2^{-(\secp+n_\secp)}.
        \end{align}
By taking the union bound over $x\in \bit^{n_\secp}$ it implies \Cref{eq:fixed_randomness_correctness}.
Moreover, the modification of $\Obf$ and $\Eval$ does not affect the security. Thus, this completes the proof of \Cref{lem:fixed_randomness_correctness}.
\end{proof}
Then we prove \Cref{thm:QCC_QCCC_PKE}. 
\begin{proof}[Proof of \Cref{thm:QCC_QCCC_PKE}]
Let $(\Obf,\Eval)$ be a $(\QQ,\CC,\CC)$-iO for classical circuits.
By \Cref{lem:fixed_randomness_correctness}, we can assume that it satisfies fixed randomness correctness without loss of generality. 
Let $m$ and $s$ be polynomials as in \Cref{thm:main}.
Then we construct a QCCC PKE scheme $(\Gen,\Enc,\Dec)$ as follows: 
    \begin{itemize}
        \item $\Gen(1^\secp)$: 
        Choose $k\gets \bit^{m(\secp)}$, and compute  $\hat{P}_{k,s(\secp)}\gets\Obf(1^\secp,P_{k,s(\secp)})$, where we recall that $P_{k,s(\secp)}$ denotes the canonical circuit of size $s(\secp)$ for the point function with target $k$. Output the 
        classical public key $\pk\coloneqq(1^\secp,\hat{P}_{k,s(\secp)})$ and 
        the classical secret key $\sk\coloneqq k$.
        \item $\Enc(\pk=(1^\secp,\hat{P}_{k,s(\secp)}),\msg)$: 
        Choose $r\gets \mathcal{R}_\secp$ where $\mathcal{R}_\secp$ is the randomness space of $\Eval(\hat{P}_{k,s(\secp)},k')$ for $k'\in \bit^{m(\secp)}$. 
        Let $C[\hat{P}_{k,s(\secp)},\msg,r]$ be a classical circuit that takes $k'\in \bit^{m(\secp)}$ as input and outputs $\msg$ if $\Eval(\hat{P}_{k,s(\secp)},k';r)=1$ and $0$ otherwise. 
        Compute $\hat{C}[\hat{P}_{k,s(\secp)},\msg,r]\gets\Obf(1^\secp,C[\hat{P}_{k,s(\secp)},\msg,r])$. 
        Output the ciphertext $\ct=\hat{C}[\hat{P}_{k,s(\secp)},\msg,r]$. 
        \item $\Dec(\sk=k,\ct=\hat{C}[\hat{P}_{k,s(\secp)},\msg,r])$: Compute $\msg'\gets \Eval(\hat{C}[\hat{P}_{k,s(\secp)},\msg,r],k)$ and output $\msg'$. 
    \end{itemize}
    By the correctness of  $(\QQ,\CC,\CC)$-iO, the above scheme clearly satisfies the correctness of PKE. 

    Below, we prove that it satisfies the IND-CPA security. For any QPT adversary $\cA$ and $b\in \bit$, we consider the following hybrid experiments:
    \begin{description}
    \item[$H_{1,b}$:] This is the original IND-CPA security experiment. That is, it works as follows: 
        \begin{enumerate}
            \item The challenger generates 
            $\hat{P}_{k,s(\secp)}\gets\Obf(1^\secp,P_{k,s(\secp)})$ for $k\gets \bit^{m(\secp)}$, and sends $\pk=(1^\secp,\hat{P}_{k,s(\secp)})$ to $\cA$. 
            \item $\cA$ chooses $\msg_0,\msg_1\in(\bit^*)^2$ of the same length and sends them to the challenger.
            \item The challenger generates $\hat{C}[\hat{P}_{k,s(\secp)},\msg_b,r]\gets\Obf(1^\secp,C[\hat{P}_{k,s(\secp)},\msg_b,r])$ 
            for $r\gets \mathcal{R}_\secp$,  
            and sends $\ct_b=\hat{C}[\hat{P}_{k,s(\secp)},\msg_b,r]$ to $\cA$.
            \item $\cA$ outputs $b'$, which is the output of the experiment.
        \end{enumerate}
        Our goal is to prove that $|\Pr[H_{1,0}=1]-\Pr[H_{1,1}=1]|\le \negl(\secp)$.
         \item[$H_{2,b}$:] This is identical to $H_{1,b}$ except that $\pk$ is set to be $\hat{Z}_{m(\secp),s(\secp)}\gets\Obf(1^\secp,Z_{m(\secp),s(\secp)})$, 
         and consequently $\ct_b$ is set to be  $\hat{C}[\hat{Z}_{m(\secp),s(\secp)},\msg_b,r]\gets\Obf(1^\secp,C[\hat{Z}_{m(\secp),s(\secp)},\msg_b,r])$
         where we recall that $Z_{m(\secp),s(\secp)}$ denotes the canonical zero-function on $m(\secp)$-bit inputs of size $s(\secp)$. By a straightforward reduction to \Cref{thm:main} for the case of $\ell=1$, we have 
         $|\Pr[H_{1,b}=1]-\Pr[H_{2,b}=1]|\le \negl(\secp)$ for $b\in \bit$.  
          \item[$H_{3,b}$:] This is identical to $H_{2,b}$ except that $\ct_b$ is set to be $\hat{Z}_{m(\secp),s'(\secp)}\gets\Obf(1^{\secp},Z_{m(\secp),s'(\secp)})$, where $s'(\secp)$ is the size of $C[\hat{Z}_{m(\secp),s(\secp)},\msg_b,r]$. (We assume without loss of generality that the size only depends on $\secp$ by padding.)
          By the fixed randomness correctness of the $(\QQ,\CC,\CC)$-iO, $C[\hat{Z}_{m(\secp),s(\secp)},\msg_b,r]$ is functionally equivalent to $Z_{m(\secp),s'(\secp)}$ with overwhelming probability over the choice of $r$. 
          Thus, by a straightforward reduction to the security of $(\QQ,\CC,\CC)$-iO, 
           we have 
         $|\Pr[H_{2,b}=1]-\Pr[H_{3,b}=1]|\le \negl(\secp)$ for $b\in \bit$.  

         Moreover, in $H_{3,b}$, no information of $b$ is given to $\cA$, and thus $\Pr[H_{3,0}=1]=\Pr[H_{3,1}=1]$. 
        \end{description}
        Combining the above, we obtain $|\Pr[H_{1,0}=1]-\Pr[H_{1,1}=1]|\le \negl(\secp)$. This completes the proof of IND-CPA security. 
\end{proof}
Since IND-CPA secure QCCC PKE implies EV-OWPuzz, which in turn implies OWPuzz, OWSGs, QEFID pairs, and EFI pairs~\cite{STOC:KhuTom24,C:ChuGolGra24}, we have the following corollary.  
\begin{corollary}\label{cor:QCC_EVOWPuzz}
    Suppose $\mathsf{NP}\nsubseteq\mathsf{i.o.BQP}$ and there exists $(\QQ,\CC,\CC)$-iO for classical circuits.  Then there exist EV-OWPuzz, OWPuzz, OWSGs, QEFID pairs, and EFI pairs. 
\end{corollary}


\subsection{C-Obf, Q-Eval, and C-Encoding}
Here, we study implications of $(\CC,\QQ,\CC)$-iO, where $\Obf$ is a classical algorithm, $\Eval$ is a  quantum algorithms, and an obfuscated encoding $\hat{C}$ is classical. 
We prove the following theorem:
\begin{theorem}\label{thm:CQC_OWF_QCCC_PKE}
    Suppose $\mathsf{NP}\nsubseteq\mathsf{i.o.BQP}$ and there exists $(\CC,\QQ,\CC)$-iO for classical circuits. Then there exist OWFs and an IND-CPA secure QCCC PKE scheme. 
\end{theorem}
To show \cref{thm:CQC_OWF_QCCC_PKE}, we rely on the following lemma:
\begin{lemma}[\cite{Gol90}]
\label{lem:EFID_OWF}
    The following two conditions are equivalent:
    \begin{itemize}
        \item There exist OWFs.
        \item There exist pairs of classical-polynomial-time-samplable distributions that are statistically far but computationally indistinguishable.
    \end{itemize}
\end{lemma}
\begin{lemma}[\cite{STOC:SahWat14}]\label{lem:PKE}
    If there exist $(\CC,\CC,\CC)$-iO for classical circuits and OWFs, then there exist IND-CPA secure PKE schemes.
\end{lemma}
In the construction of \cite{STOC:SahWat14}, the encryption algorithm runs $\Eval$.
Thus, by adapting their construction to $(\CC,\QQ,\CC)$-iO in which only $\Eval$ is quantum algorithm but both of $\Obf$ and the obfuscated encoding are classical, we obtain QCCC PKE scheme in which only the encryption algorithm is quantum.
\begin{corollary}\label{cor:QCCCPKE}
    If there exist $(\CC,\QQ,\CC)$-iO for classical circuits and OWFs, then there exist IND-CPA secure QCCC PKE schemes.
\end{corollary}
Now we are ready to prove \cref{thm:CQC_OWF_QCCC_PKE}.
\begin{proof}[Proof of \cref{thm:CQC_OWF_QCCC_PKE}]
    By \cref{lem:EFID_OWF,cor:QCCCPKE}, it suffices to construct a pair of classical-polynomial-time-samplable distributions that are statistically far but computationally indistinguishable.
    By applying \cref{thm:main} for $\ell=1$, there exist classical-polynomial-time-computable polynomials $m$ and $s$ such that the following two distributions are computationally indistinguishable:
    \begin{itemize}
        \item $\cD_0(\secp)$: Sample $k\gets\bit^{m(\secp)}$, run $\hat{P}_{k,s(\secp)}\gets\Obf(1^\secp,P_{k,s(\secp)})$, and output $\hat{P}_{k,s(\secp)}$.
        \item $\cD_1(\secp)$: Run $\hat{Z}_{m(\secp),s(\secp)}\gets\Obf(1^\secp,Z_{m(\secp),s(\secp)})$ and output $\hat{Z}_{m(\secp),s(\secp)}$.
    \end{itemize}
    Moreover both of $\cD_0(\secp)$ and $\cD_1(\secp)$ are classical-polynomial-time-samplable because $m$ and $s$ are classical-polynomial-time-computable and $\Obf$ is a PPT algorithm.
    Thus, to complete the proof, we show that $\cD_0(\secp)$ and $\cD_1(\secp)$ are statistically far.
    To show this, we construct an unbounded-time distinguisher $\cA$ that distinguishes $\cD_0(\secp)$ and $\cD_1(\secp)$. 
    \begin{description}
        \item[$\cA(\hat{C})$:]
        Upon receiving an obfuscated encoding $\hat{C}$,  
        it computes $p_{k'}:=\Pr[\Eval(\hat{C},k')=1]$  for all $k'\in \bit^{m(\secp)}$. If there is $k'\in \bit^{m(\secp)}$ such that 
        $p_{k'}\ge 1/2$, it outputs $0$, otherwise it outputs $1$.  
    \end{description}
If $\hat{C}=\hat{P}_{k,s(\secp)}\gets \cD_0(\secp)$, the correctness of the iO implies that $\Pr[p_k\ge 1-\negl(\secp)]\ge 1-\negl(\secp)$ where the probability is taken over the randomness in the sampling procedure of $\cD_0$. Thus, we have $\Pr[\cA(\hat{C})=0]\ge 1-\negl(\secp)$. On the other hand, if $\hat{C}=\hat{Z}_{m(\secp),s(\secp)}\gets \cD_1(\secp)$, the correctness of the iO implies that $\Pr[\forall k'\in \bit^{m(\secp)},  p_{k'}\le \negl(\secp)]\ge 1-\negl(\secp)$ where the probability is taken over the randomness in the sampling procedure of $\cD_1$. Thus, we have $\Pr[\cA(\hat{C})=1]\ge 1-\negl(\secp)$. Therefore, $\cD_0(\secp)$ and $\cD_1(\secp)$ are statistically far and we complete the proof.
\end{proof}

\begin{remark}
    One might think that \Cref{thm:CQC_OWF_QCCC_PKE} directly follows from an adaptation of the technique of \cite{FOCS:KMNPRY14} since we can derandomize $\Obf$ when it is classical. In fact, this is true in the perfectly correct case. On the other hand, this does not work in the imperfect case (as in \Cref{def:qiO}) since their proof in the imperfect setting involves obfuscation of an obfuscated circuit, but this is not possible in our setting since $\Eval$ is a quantum algorithm and thus cannot be obfuscated by iO for classical circuits as is considered in this paper.  
\end{remark}

Note that $\mathsf{NP}\nsubseteq\mathsf{i.o.BQP}$ and the existence of $(\CC,\QQ,\CC)$-iO for classical circuits imply all Microcrypt primitives since they imply the existence of OWFs (and therefore PRUs~\cite{cryptoeprint:2024/1652}). 

\subsection{C-Obf, C-Eval, and C-Encoding}
Here, we study implications of $(\CC,\CC,\CC)$-iO, where both $\Obf$ and $\Eval$ are classical algorithms and an obfuscated encoding $\hat{C}$ is classical. This is oftend referred to as post-quantum iO. 
We prove the following theorem:
\begin{theorem}\label{thm:CCC_OWF_PKE}
    Suppose $\mathsf{NP}\nsubseteq\mathsf{i.o.BQP}$ and there exists $(\CC,\CC,\CC)$-iO for classical circuits. Then there exist OWFs and an IND-CPA secure PKE scheme. 
\end{theorem} 
\begin{proof}[Proof of \cref{thm:CCC_OWF_PKE}]
    This proof is essentially same as the proof of \cref{thm:CQC_OWF_QCCC_PKE}.
    By \cref{lem:EFID_OWF,lem:PKE}, it suffices to construct a pair of classical-polynomial-time-samplable distributions that are statistically far but computationally indistinguishable.
    By applying \cref{thm:main} for $\ell=1$, there exist classical-polynomial-time-computable polynomials $m$ and $s$ such that the distributions over $\Obf(1^\secp,P_{k,s(\secp)})$ and $\Obf(1^\secp,Z_{m(\secp),s(\secp)})$ are computationally indistinguishable, where $k\gets\bit^{m(\secp)}$. We can show that they are classical-polynomial-time-samplable and statistically far by adapting the same argument used in the proof of \cref{thm:CQC_OWF_QCCC_PKE}.
\end{proof}

\begin{remark}
We could also prove \Cref{thm:CCC_OWF_PKE} by a straightforward adaptation of \cite{FOCS:KMNPRY14}. On the other hand, an advantage of our approach is that it in fact only needs iO for 3CNF formulas rather than general classical circuits. 
Constructing OWFs from imperfect iO for 3CNF formulas was an open problem left by \cite{FOCS:KMNPRY14}, and our alternative proof resolves this open problem, though the open problem itself was also recently resolved (in a stronger form) in \cite{STOC:HirNan24,cryptoeprint:2024/800} by completely different techniques. 
\end{remark}

Note that $\mathsf{NP}\nsubseteq\mathsf{i.o.BQP}$ and the existence of $(\CC,\CC,\CC)$-iO for classical circuits imply all Microcrypt primitives since they imply the existence of OWFs. 

\ifnum\anonymous=1
\else
{\bf Acknowledgements.}
We thank Minki Hhan, Giulio Malavolta, and Kabir Tomer for insightful discussions on the initial ideas of this work during QIP 2024. 
TM is supported by
JST CREST JPMJCR23I3,
JST Moonshot R\verb|&|D JPMJMS2061-5-1-1, 
JST FOREST, 
MEXT QLEAP, 
the Grant-in Aid for Transformative Research Areas (A) 21H05183,
and 
the Grant-in-Aid for Scientific Research (A) No.22H00522.
YS is supported by JST SPRING, Grant Number JPMJSP2110.
\fi

\appendix

\ifnum\submission=0
\bibliographystyle{alpha} 
\else
\bibliographystyle{splncs04}
\fi
\bibliography{abbrev3,crypto,reference,text}

\end{document}